\documentclass[11pt, one sided]{amsart}
\RequirePackage[colorlinks,citecolor=blue,urlcolor=blue]{hyperref}

\usepackage{latexsym,amssymb,amsmath,amsfonts,amsthm}
\usepackage{amsthm,amsmath}
\usepackage{float}
\usepackage{enumerate}
\usepackage[margin=3cm]{geometry}
\usepackage{bbm}
\usepackage{subfigure}
\usepackage{graphicx}
\usepackage[toc,page]{appendix}
\usepackage{multirow}
\usepackage{tikz}

\usepackage{comment}

\usepackage{amsfonts}
\usepackage[all]{xy}
\usepackage{caption}
 \usepackage{geometry}                % See geometry.pdf to learn the layout options. There are lots.
\usepackage{graphicx}
\usepackage{amssymb}
\usepackage{epstopdf}
\usepackage{color}
\usepackage{bbm, dsfont}
\usepackage{hyperref}
\usepackage[utf8]{inputenc}
\usepackage{amssymb,amsthm,amsmath,amssymb,wrapfig,dsfont}
\usepackage{tikz}
\usetikzlibrary{decorations.pathmorphing}
\tikzset{snake it/.style={decorate, decoration=snake}}
\usetikzlibrary{shapes.geometric,positioning,decorations.pathreplacing} 
\usepackage{pgfplots}

\newtheorem{theorem}{Theorem}

\newtheorem{lemma}{Lemma}%[section]

\def\beq{ \begin{equation} }
\def\eeq{ \end{equation} }

\def\square{\vcenter{\vbox{\hrule height .4pt
  \hbox{\vrule width .4pt height 5pt \kern 5pt
        \vrule width .4pt} \hrule height .4pt}}}

\newcommand*{\xMin}{0}%
\newcommand*{\xMax}{30}%
\newcommand*{\yMin}{0}%
\newcommand*{\yMax}{30}%
\newcommand*{\xMiN}{0}%
\newcommand*{\xMaX}{3}%
\newcommand*{\yMiN}{0}%
\newcommand*{\yMaX}{3}%

\newcommand{\bae}{\begin{equation}\begin{aligned}}
\newcommand{\eae}{\end{aligned}\end{equation}}

\DeclareFontFamily{OML}{rsfs}{\skewchar\font'177}
\DeclareFontShape{OML}{rsfs}{m}{n}{ <5> <6> rsfs5 <7> <8> <9>
rsfs7 <10> <10.95> <12> <14.4> <17.28> <20.74> <24.88> rsfs10 }{}
\DeclareMathAlphabet{\mathfs}{OML}{rsfs}{m}{n}

\newcommand{\note}[1]{{\color{red}{ \bf{ [Note: #1]}}}}

\usepackage[nomessages]{fp}

\begin{document}
\pgfmathsetseed{2}
\title{Relativistic Propagators on Lattices}

\author{Rory O'Dwyer}
\address[Rory O'Dwyer]{Department of Physics, Stanford University}
%\email{rodwyer@stanford.edu}

\maketitle
\begin{abstract}

I define the lattice propagator on a very general collection of graphs, namely graphs locally isomorphic to $\mathbb{Z}^{d}\times \mathbb{Z}$. I then define polygonal approximations to the minkowski metric and define a corresponding lattice propagator for these. I show in $d=1$, as suggested by the metric approximation, the continuum limit of the polygonal propagators converges to the Klien Gordon Propagator. Finally, I obtain the taxicab polygonal propagator in a very general collection of spaces, including $\mathbb{T}^{d}$, the Klein bottle, and a discretization of de-Sitter space.

\end{abstract}

%\tableofcontents

%***********************************************************************************************************************************************************
%***********************************************************************************************************************************************************

\section{Introduction}
\label{section:intro}

This paper aims to define the propagator of QFT (Quantum Field Theory) in terms of geometrical quantities, as per the original intuition of Feynman. Feynman's original derivation of the path integral formulation of the propagator expressed it in terms of a limit of oscillatory integrals \cite{feynman_hibbs_styer_2017}, but there are mathematical problems in making this expression rigorous in its most general setting.  The path integral as a measure on paths works in the setting of Statistical Mechanics in the Kac-Feynman formula \cite{glimm_jaffe_1987}, but attempts to do the same in the setting of propagators have lead to candidate measures that are not $\sigma$-additive \cite{Albeverio2009}. Other attempts to make the path integral rigorous include treating it as a functional that converges on a family of well-behaved functions \cite{Albeverio2009} and using Gaussian free fields to define the propagator in the setting of Conformal Field Theory \cite{https://doi.org/10.48550/arxiv.1101.1024}. Most approaches are highly abstracted from treating the path integral as a phase-weighted sum over paths in space. This paper aims to develop another approach that centrally employs this intuition.

In Section~\ref{section:notation}, I define a discrete propagator as a complex-valued function on pairs of vertices in a discretization of space. The function itself is a sum over achronal local paths (defined in Section~\ref{section:notation}) weighted by a phase whose angle is the minkowski length of those paths. Because these quantities are finite, I don't need to worry about convergence issues that plague most formalizations of propagators. I  obtain analytic expressions for these propagators and  then show that these analytic expressions converge to a continuum limit as ultra-fine lattice lengths are considered. This employs recent work in continuous lattice path counting in \cite{continuousbin} and \cite{continuouslatticepath} where a continuous multinomial coefficient is defined (another result of this work is that this coefficient is a well defined, fully differentiable function). To show that the objects I obtain are indeed discretizations of the QFT flat-space free scalar propagator, I use transformation invariance arguments. I show that any sup norm limit of lattice path integrals (as more directions on the lattice that paths are allowed to occupy are included) converge to the closed-form expression for the QFT propagator. 

I must direct some of this introduction towards polygonal minkowski metrics. Put plainly, a polygonal metric is any metric whose unit sphere is a polyhedron. We need to make use of these metrics as the action of a relativistic bosonic particle is related to an integral over a path of its normal minkowski length. We will need a discretization of the minkowski metric to define our discrete propagators.

The most well-known polygonal metric is the taxicab metric on Euclidean space; its associated perimeter shows up as the free energy of many Ising model solutions \cite{latticeProb},\cite{duncan2021discrete}. In this paper, I will show that a minkowski taxicab metric is the first among a series of metrics that converge to the minkowski one. This family of polygonal metrics $d_{n}$ has an interesting relationship to Pythagorean triples \cite{rational}, which I develop in Section~\ref{section:def}. I show for $d=2$ the continuum propagators for $d_{n}$ converge to the minkowski propagator as $n\rightarrow \infty$, thereby motivating the idea that the taxicab propagator is a 'first order' approximation to the more difficult minkowski propagator. Finally, I demonstrate the versatility of the taxicab minkowski propagator.

I show that lattice propagators can be calculated on discretizations of general orbifolds. I will also work with a tropical formulation of de-Sitter space \cite{tropicalGeometry},\cite{linde_2017}. I obtain continuum propagators for these settings. I also derive that the taxicab metric propagators on de-Sitter space are ignorant of compact dimensions, but the propagator probes these dimensions well in flat space. Finally, I mention that all the propagators I obtain in this work will correspond to free particles without interaction terms. Interactions are difficult to express geometrically without more complicated theories (like String theory \cite{zwiebach_2004}). Generalizing the taxicab propagators in the setting of strings and membranes will be the subject of latter work.

%I will show that the discrete propogator generalizes to arbitrary dimension and a general class of orbifolds (\note{cite} for info on orbifolds), and I find an analogous discrete taxicab operator whose continuum limit is easy to obtain in many settings because of multinomial concentration inequalities \note{cite}.

\section{Notations and Statement of Theorems}
\label{section:notation}
Let $d\in\{1,...,\}$ and $\{L_{i}\}_{i=1}^{d}\subset \mathbb{N}\cup\{\infty\}$. I define $X=(\times_{i}\{-L_{i},-L_{i}+1,...,L_{i}\})\times \mathbb{Z}$. The first $d$ copies of $\mathbb{Z}$ are defined as the \textbf{spatial} coordinates and the last with \textbf{time}. I note that $X\subset \mathbb{Z}^{d}$, so it has a boundary $\partial X$ in $\mathbb{Z}^{d}$. I will constrain ourselves to $X$ such that $2|\left|\partial X\right|)$. I can therefore give $\partial X$ an arbitrary equivalence relation $\sim$ which partitions $\partial X$ into pairs. I shall leave $\sim$ completely general. I may also consider drawing X from a tropical polynomial context. Let $\{\vec{a}^{i}\}_{i=1}^{n}\subset \mathbb{Z}^{d}\times \mathbb{Z},\{b\}_{i=1}^{n}\subset \mathbb{Z},c\subset \mathbb{Z}$. Then, I define an arbitrary \textbf{tropical polynomial} $\mathfrak{p}:\mathbb{Z}^{d}\times\mathbb{Z}\rightarrow \mathbb{Z}$ as the following:

\begin{equation}
    \label{eq:trop}
    \mathfrak{p}(\vec{x})=max_{i\in\{1,...,n\}}(\vec{a}*\vec{x}+b)-c
\end{equation}

I will occasionally let $X=\{\vec{x}\in \mathbb{Z}^{d}\times\mathbb{Z}|\mathfrak{p}(\vec{x})=0\}$. In the case of these X's I may have to specify what coordinates are time and spatial, though unless explicitly stated, it will be that convention induced by $\mathbb{Z}^{d}\times\mathbb{Z}$. Our definitions will generally hold for all the above-described $X$.

If $\vec{x}\in X$, then $proj_{t}(\vec{x})$ is the projection of $\vec{x}$ into the time coordinate (and I will denote as $proj_{x_{i}}(\vec{x})$  any space coordinate projections). I give $X$ the graph structure $(\mathcal{V},\mathcal{E})$ where $\mathcal{V}$ are its points and $\mathcal{E}=\{\vec{x}_{1}\vec{x}_{2}|\vec{x}_{1},\vec{x}_{2}\in X,|\vec{x}_{1}-\vec{x}_{1}|_{l_{\infty}}=1\}$. This has a metric given to it by the graph distances. Let $\vec{x}_{1},\vec{x}_{2}\in X$. I define on the following functions $d_{l_{2}},d_{l_{2}^{*}}:X\times X\rightarrow \mathbb{C}$

\begin{multline}
    \label{eq:minko}
    d_{l_{2}}(\vec{x}_{1},\vec{x}_{2})=\sqrt{proj_{t}(\vec{x}_{1}-\vec{x}_{2})^{2}+\sum_{i=1}^{d}proj(\vec{x}_{1}-\vec{x}_{2})^{2}}
    \\
    d_{l_{2}^{*}}(\vec{x}_{1},\vec{x}_{2})=\sqrt{proj_{t}(\vec{x}_{1}-\vec{x}_{2})^{2}-\sum_{i=1}^{d}proj(\vec{x}_{1}-\vec{x}_{2})^{2}}
\end{multline}

For all our purposes, I will only be using $d_{l_{2}^{*}}$ on $d_{l_{2}^{*}}^{-1}(\{a|a\in \mathbb{R},a\ge 0\})$. The first equation in Equation~\ref{eq:minko} is clearly the $l_{2}$ metric on $X$. The subscript $l_{2}^{*}$ in the second is used to denote its similarity with the $l_{2}$ metric; it is the analog of the minkowski metric in our setting. I will therefore refer to the second part of Equation~\ref{eq:minko} as the \textbf{minkowski metric}. Let $n\in \{1,2,...\}$. Consider the set of scaled primitive \textbf{pythagorean triples} with hypotenuse below n, namely $\mathcal{A}_{n}=\{(\frac{x}{I},\frac{t}{I})\in\mathbb{Q}^{2}|x^{2}+I^{2}=t^{2},(x,I,t)\in\mathbb{Z}^{3},0\le t\le n\}$. The word primitive places an extra constraint on $\mathcal{A}_{n}$ such that if we consider the equivalence class placed upon these vectors by parallelism, then we only keep the representative with the lowest hypotenuse. These points have a natural ordering from least x coordinate to most; let us denote $\vec{a}_{i}\in\mathcal{A}_{n}$ as the i-th vector under this ordering. For $n\ge 2$, then, we may define Equation~\ref{eq:polygon} as our \textbf{polygonal minkowski metric}.

%\begin{multline}
    %\label{eq:polygon}
    %d_{n}(\vec{x}_{1},\vec{x}_{2})=min_{k\in\{0,...,k-1\}}(|\frac{2kn-2(k+1)^{2}-n^{2}+n}{2k(k-n+2)-3n+1}proj_{t}(\vec{x}_{2}-\vec{x}_{1})|
%    \\-|\frac{(n+1)(2k-n+2)}{2k(k-n+2)-3n+1}proj_{x}(\vec{x}_{2}-\vec{x}_{1})|)
%\end{multline}

\tiny
\begin{multline}
    \label{eq:polygon}
    d_{n}(\vec{x}_{1},\vec{x}_{2})=min_{i\in\{0,...,|\mathcal{A}_{n}|-1\}}(\frac{|proj_{x}(\vec{a}_{i+1}-\vec{a}_{i})|proj_{t}(\vec{x}_{2}-\vec{x}_{1})|-|proj_{t}(\vec{a}_{i+1}-\vec{a}_{i})proj_{x}(\vec{x}_{2}-\vec{x}_{1})|}{|proj_{x}(\vec{a}_{i+1}-\vec{a}_{i})proj_{t}(\vec{a}_{i})|-|proj_{t}(\vec{a}_{i+1}-\vec{a}_{i})proj_{x}(\vec{a}_{i})|})
    \\
    \wedge (\frac{|proj_{t}(\vec{x}_{2}-\vec{x}_{1})|-|proj_{x}(\vec{x}_{2}-\vec{x}_{1})|}{|proj_{t}(\vec{a}_{0})|-|proj_{x}(\vec{a}_{0})|})
\end{multline}
\normalsize

For $n=1$, we define $d_{1}(\vec{x}_{2}-\vec{x}_{1})$ as $|proj_{t}(\vec{x}_{2}-\vec{x}_{1})|-|proj_{x}(\vec{x}_{2}-\vec{x}_{1})|$. These are functions $X\times X\rightarrow \{a|a\in \mathbb{R},a\ge 0\}$, and they are meant to be natural polygonal approximations of the minkowski metric. One can see this via Figure~\ref{fig:CloserCirc}. When $n=1$, this becomes $|x|-|y|$, or the taxicab minkowski metric. I can extend this metric to high dimensions as Ill in Equation~\ref{eq:taxicab}.

\begin{equation}
    \label{eq:taxicab}
    d_{l_{1}}(\vec{x}_{1},\vec{x}_{2})=|proj_{t}(\vec{x}_{2}-\vec{x}_{1})|-\sum_{i=1}^{d}|proj(\vec{x}_{1}-\vec{x}_{2})|
\end{equation}

I shall return back to Equation~\ref{eq:minko} and Equation~\ref{eq:polygon} for a moment. Let's define $\mathcal{A}^{all}=\{\vec{x}\in \mathbb{Z}\times\mathbb{Z}|d_{l_{2}^{*}}(0,\vec{x})=d_{n}(0,\vec{x})\ge 0,proj_{t}(\vec{x})>0\}$. A subset $\mathcal{A}^{gen}\subset \mathcal{A}^{all}$ is called generating iff for all $a\in \mathcal{A}^{all}$ there is $\{a_{i}\}_{i=1}^{N}\subset \mathcal{A}^{gen}$ such that $\sum_{i=1}^{N}=a$. Then an \textbf{axes of symmetry} of $d_{n}$, denoted $\mathcal{A}_{n}$, is some $\mathcal{A}_{gen}$ such that if $\vec{x}_{1},\vec{x}_{2}\in \mathcal{A}_{n}$ and $\vec{x}_{1}\neq c\vec{x}_{2}$ then $c=1$. So, every direction has a unique representative. The double usage of $\mathcal{A}_{n}$ for scaled Pythagorean triples and as axes of symmetry of $d_{n}$ will be shown to not be an abuse of notation in Theorem~\ref{thm:axes} (except the axes of symmetry of $d_{n}$ will include also light paths). Our language in this paper will suggest that this set is unique to each $d_{n};$ we will not prove this for we do not require it for the next concepts.

%$\in argmin(\{\left|A^{gen}\right|\textrm{ such that }A^{gen}\subset A^{cal}\})$.

%\{d_{l_{2}^{*}}(\vec{x})|\vec{x}=n\vec{x}'\})|\vec{x}'\in \mathcal{A},n\in \mathbb{N}\}

I define as a \textbf{path} as set $\gamma=\{\vec{x}_{i}\}_{i=1}^{n}$ such that $\{x_{i}x_{i+1}\}_{i=1}^{n-1}\subset\mathcal{E}$. Then, our \textbf{achronal, local paths} between $\vec{x}$ and $\vec{y}\in X$ as the following:

\begin{equation}
    \label{eq:pathset}
    \Gamma^{\vec{x},\vec{y}}=\{\gamma|\gamma|_{0}=\vec{x},\gamma|_{n}=\vec{y},d_{l_{2}^{*}}(\vec{x}_{i},\vec{x}_{i+1})\ge 0,d_{l_{2}^{*}}(\vec{x},\vec{y})\in \mathbb{Z},proj_{t}(\vec{x}_{i+1},\vec{x}_{i})\ge 0\}
\end{equation}

I define our \textbf{achronal, local $n$ polygonal paths} between $\vec{x}$ and $\vec{y}\in X$ as the following:
\small
\begin{equation}
    \label{eq:pathset2}
    \Gamma_{n}^{\vec{x},\vec{y}}=\{\gamma|\gamma|_{0}=\vec{x},\gamma|_{n}=\vec{y},d_{l_{2}^{*}}(\vec{x}_{i},\vec{x}_{i+1})=d_{n}(\vec{x}_{i},\vec{x}_{i+1})\ge 0,d_{n}(\vec{x}_{i},\vec{x}_{i+1})\in\mathbb{Z},proj_{t}(\vec{x}_{i+1},\vec{x}_{i})\ge 0\}
\end{equation}
\normalsize
On $\Gamma_{n}$ and $\Gamma$ we place an equivalence relation on paths; two paths are equivalent if their indexing set traces out the same piecewise linear paths in $X$. The use of axes of symmetry is made clear in this context; each equivalence class in $\Gamma_{n}^{\vec{x},\vec{y}}$ has a unique representative with difference sequence drawn from $\mathcal{A}_{n}$ (Theorem~\ref{thm:generate}).

Clearly, from Equation~\ref{eq:pathset} and Equation~\ref{eq:pathset2} I have $\Gamma^{\vec{x},\vec{y}}_{n}\subset \Gamma^{\vec{x},\vec{y}}$. Let $\gamma\in \Gamma^{\vec{x},\vec{y}}$. Then I define the polygonal and minkowski \textbf{proper time} in Equation~\ref{eq:inter}.

\begin{equation}
    \label{eq:inter}
    \rho_{n}(\gamma)=\sum_{i=0}^{n-1}d_{n}(x_{i+1},x_{i}),\rho_{l_{2}^{*}}(\gamma)=\sum_{i=0}^{n-1}d_{l_{2}^{*}}(x_{i+1},x_{i})
\end{equation}

From this definition, I immediately obtain Theorem~\ref{thm:axes}. These geometrical results are developed because I will be viewing the path integral as a geometrical object. For this purpose, I define the following functions $K_{n}:X\times X\rightarrow \mathbb{C}$ and $K_{l_{2}^{*}}:X\times X\rightarrow \mathbb{C}$. Let $m\in \{a|a\in \mathbb{R},a>0\}$.

\begin{equation}
    \label{eq:propogators}
    K_{n}(\vec{x},\vec{y})=\sum_{\gamma\in\Gamma_{n}^{\vec{x},\vec{y}}}e^{im\rho_{n}(\gamma)},
    K_{l_{2}^{*}}(\vec{x},\vec{y})=\sum_{\gamma\in\Gamma^{\vec{x},\vec{y}}}e^{im\rho_{l_{2}^{*}}(\gamma)}
\end{equation}

These functions are our \textbf{discrete propagators}. They are the main subject of study in this paper.

These are all our basic definitions in the discrete context. I will devote some of our paper to the continuous setting. For this, I must define for our original X a natural domain $X^{cont}$, which is the result of taking infinitely fine lattices. In the first context (non tropical), it's $X^{cont}=(\times_{i}-[L_{i},L_{i}])\times \mathbb{R}$ with an associated equivalence relation $\sim$ on pairs of points in $\partial X^{cont}$ as a subspace of $\mathbb{R}^{d}\times \mathbb{R}$. I note that if $\mathfrak{p}$ is the definitional tropical polynomial of $X$, then $\mathfrak{p}$ can be immediately extended to $\vec{x}\in \mathbb{R}^{d}\times\mathbb{R}$. $X^{cont}$ will then be the zero set of this extended $\mathfrak{p}$. I now want a generalization of the closest integer function $[*]:X^{cont}\rightarrow X$. For $L\in \mathbb{N}\cup\{\infty\}$ I let $[*]^{1d}:[-L,L]\rightarrow \{-L,-L+1,...,L\}$ be the closest point in the range to the domain under the $l_{2}$ metric. 

Then our generalization is defined as in Equation~\ref{eq:closest}.

\begin{multline}
    \label{eq:closest}
    [\vec{x}]\textrm{ is the element of }X\textrm{ such that }
    \\
    proj_{x_{i}}([\vec{x}])=[proj_{x_{i}}(\vec{x})]^{1d}
    \\
    proj_{x_{i}}([\vec{x}])=[proj_{t}(\vec{x})]^{1d}
\end{multline}

Finally, for those $X$ which have some finite element $L_{i}$ in their definition, I must also change their dimensions to approach $X^{cont}$ as a limiting space for infinitely fine lattices. Therefore, when I write $\Gamma_{n}|_{m},\Gamma|_{m},K_{n}|_{m}$, or $K_{l_{2}^{*}}|_{m}$, I am referring to those objects for $X|_{m}=(\times_{i}\{-mL_{i},-mL_{i}+1,...,mL_{i}\})\times \mathbb{Z}$ when $X=(\times_{i}\{-L_{i},-L_{i}+1,...,L_{i}\})\times \mathbb{Z}$. 

We will find that recent developments in combinatorics, mentioned in Section~\ref{section:intro}, facilitate the expression of continuum propagators in all of these settings. Counting lattice paths are naturally connected to the multinomial coefficient, and therefore continuum propagators seem to require the development of a continuous multinomial coefficient. As defined by Cano and Diaz in \cite{continuousbin}, and later generalized by Wakhare, Vignat, Le, and Robins in \cite{continuouslatticepath}, we consider the continuous multinomial coefficient in Equation~\ref{eq:contmult} for l variables. Let $\{x_{i}\}_{i=1}^{l}\subset \mathbb{R}_{+}$, $\mu$ denote the Lebesgue measure on $\mathbb{R}^{n}$ \cite{folland}, and for $n,N\in \mathbb{N}$ denote  by $D(n,N)$ the number of Smirnov words of length n and N \cite{continuouslatticepath}. Furthermore let $c\in D(n,N)$, let $\{c_{k}\}_{k=1}^{n}$ denote singular letters of our Smirnov word taking elements among some collection of d dimensional vectors, and let $\{e_{k}\}_{k=1}^{l}$ denote these vectors. For $q\in \mathbb{R}_{+}^{d}$ we define the \textbf{path polytope} $P(q,c)$ in Equation~\ref{eq:pathpoly}.

\begin{equation}
    \label{eq:pathpoly}
    P(q,c)=\{(\lambda_{1},...,\lambda_{n})\in \mathbb{R}_{+}^{n}|\sum_{k=0}^{n}\lambda_{k}e_{c_{k}}=q\}
\end{equation}

We now have our desired definition of the continuous multinomial in Equation~\ref{eq:contmult}:

\begin{equation}
    \begin{Bmatrix}\sum x_{i}\\x_{1},x_{2},...,x_{l}\end{Bmatrix}=\sum_{n=0}^{\infty}\sum_{c\in D(n,l)}\mu (P(q,c))
    \textrm{ where }q=\sum x_{i}e_{i}
    \label{eq:contmult}
\end{equation}

Problematically, the continuous multinomial was never shown by either \cite{continuousbin} or \cite{continuouslatticepath} to converge except in special cases. The necessary convergence results are developed in Section~\ref{section:def}; we will also show more rigorously in what sense they are continuous analogs of the discrete multinomial coefficients. Namely, for $m\in \mathbb{N}$ we introduce the operator $\mathcal{T}^{m}_{cont}$ in Equation~\ref{eq:operator}. This operator acts on sums indexed over paths in $\Gamma^{\vec{y},\vec{x}}$. Let $f:\Gamma^{\vec{y},\vec{x}}\rightarrow \mathbb{C}$ be some path indexed complex valued function and let $r,\theta:\Gamma^{\vec{y},\vec{x}}\rightarrow\mathbb{R}$ be its radial and angular components. Let $|\gamma|$ denote the number of distinct linear segments in $\gamma$. Then, we can rearrange any general sum over $\gamma\in \Gamma^{\vec{y},\vec{x}}$ as done in Equation~\ref{eq:rearrange}.

\begin{equation}
    \sum_{\gamma\in\Gamma^{\vec{x},\vec{y}}}f(\gamma)=\sum_{n=1}^{\infty}\sum_{\gamma\in\Gamma^{\vec{x},\vec{y}},|\gamma|=n}f(\gamma)
    \label{eq:rearrange}
\end{equation}

Then, we may write $\mathcal{T}_{cont}^{m}$ easily in this context (in Equation~\ref{eq:operator})

\begin{equation}
    \mathcal{T}^{m}_{cont}(\sum_{\gamma\in\Gamma^{\vec{x},\vec{y}}}f(\gamma))=\sum_{n=d}^{\infty}(\sum_{\gamma\in\Gamma^{\vec{x},\vec{y}},|\gamma|=n}r(\gamma)m^{d-n}e^{i\frac{\theta(\gamma)}{m}})
    \label{eq:operator}
\end{equation}

As demonstrated in Theorem~\ref{thm:disctocont}, this operator is required to connect a natural estimate of counting paths to a notion of volume, as one would expect in the continuous setting. With all the relevant concepts defined, the following \textbf{continuum propagators} are defined in Equation (in the event they exist)

\begin{multline}
    \label{eq:cont}
    K_{n}^{cont}(\vec{x},\vec{y})=lim_{m\rightarrow\infty}\frac{\mathcal{T}_{cont}^{m}K_{n}([m*\vec{x}],[m*\vec{y}])|_{m}}{\mathcal{T}_{cont}^{m}(max_{\vec{x'},\vec{y'}\in X|_{m},proj_{t}(\vec{y}'-\vec{x}')=[nt]}(\left|\Gamma_{n}^{\vec{x'},\vec{y'}}|_{m}\right|))}
    \\
    K_{l_{2}^{*}}^{cont}(\vec{x},\vec{y})=lim_{p\rightarrow\infty}K_{p}^{cont}(\vec{X},\vec{y})
\end{multline}

In Equation~\ref{eq:cont}, we will be satisfied if the sequence of $K_{n}^{cont}$ has a convergent subsequence towards some function in a pointwise or sup norm sense, and therefore that this limit is unique w.r.t limits. I now move on to the statement of our theorems. I aim to demonstrate that $K_{n}$ and $K_{l_{2}^{*}}$ can be defined and calculated in any of the $X$ I defined above, that $K_{n}\rightarrow K$ in some sense, and that $K_{n}^{cont}$ and $K_{l_{2}}^{cont}$ exist in a very general context and that $K_{n}\rightarrow K_{l_{2}^{*}}$.  I aim to show that $K_{l_{2}^{*}}$ can be identified with the free particle propagator from Quantum Mechanics and that this is why I should even care about $K_{n}$. 

I then show that $K_{1}$ and $K_{1}^{cont}$ has a remarkable ease of computation in all $X$. These latter theorems represent the bulk of our work; the former serve to motivate the latter. Let us first state the motivating polygonal propagator theorems, and then move on to working with $K_{1}$ and $K_{1}^{cont}$

\subsection{Statement of $K_{n}$ Theorems}

We first derive an analytic expression for $K_{d_{n}}^{cont}(\vec{x},\vec{y})$.

\begin{theorem}
\label{thm:k3}

Let $d=2$. For $\vec{x},\vec{y}\in X=\mathbb{Z}\times \mathbb{Z}$ we have

\begin{multline}
	\label{eq:k3one}
	K_{n}(\vec{y},\vec{x})=\sum_{I}(\sum_{I_{a_{1}}=0}^{C_{1}}...\textrm{no k}=\frac{|\mathcal{A}_{n}|-1}{2}...\sum_{I_{a_{|\mathcal{A}_{n}|}}=0}^{C_{|\mathcal{A}_{n}|}}\begin{pmatrix}\sum_{\vec{a}\in\mathcal{A}_{n}}I_{a}\\\Pi_{a\in\mathcal{A}_{n}}(I_{\vec{a}})\end{pmatrix})e^{mI}
	\\
	C_{j}=min(\lfloor\frac{proj_{t}(\vec{y}-\vec{x})-\sum_{i=1,i\neq\frac{|\mathcal{A}_{n}|-1}{2}}^{j-1}I_{a_{i}}proj_{t}(\vec{a}_{i})}{proj_{t}(\vec{a}_{j})}\rfloor,\lfloor\frac{I-\sum_{i=1,i\neq\frac{|\mathcal{A}_{n}|-1}{2}}^{j-1}I_{a_{i}}d_{n}(0,\vec{a})}{d_{n}(0,\vec{a}_{j})}\rfloor)
	\\
	\textrm{ where Eq~\ref{eq:propStep3} holds for $I_{\pm 0}$ and $I_{\frac{|\mathcal{A}_{n}|-1}{2}}$}
\end{multline}

and we have $K_{n}^{cont}$ equal to Equation~\ref{eq:k3two} up to a multiplicative normalization constant dependant only on t.
\tiny
\begin{multline}
    \label{eq:k3two}
    \mathcal{F}|_{I}^{m}(\int_{0}^{C_{i}}...\int_{0}^{C_{|\mathcal{A}_{timeless}|}}\begin{Bmatrix}\sum_{k\in\mathcal{A}_{timeless}}I_{k}+(I-\sum_{k\in\mathcal{A}_{timeless}}I_{k}d_{n}(0,\vec{a}_{k}))+f_{+}+f_{-}\\\{I_{k}\}_{k\in\mathcal{A}_{timeless}},I-\sum_{k\in\mathcal{A}_{timeless}}I_{k}d_{n}(0,\vec{a}_{k}),f_{+},f_{-}\end{Bmatrix}\Pi_{i\in\mathcal{A}_{timeless}}dI_{k})
    \\
    \textrm{ where }C_{j}=min(\frac{proj_{t}(\vec{y}-\vec{x})-\sum_{i=1,i\neq\frac{|\mathcal{A}_{n}|-1}{2}}^{j-1}I_{a_{i}}proj_{t}(\vec{a}_{i})}{proj_{t}(\vec{a}_{j})},\frac{proj_{t}(\vec{y}-\vec{x})-\sum_{i=1,i\neq\frac{|\mathcal{A}_{n}|-1}{2}}^{j-1}I_{a_{i}}d_{n}(0,\vec{a})}{proj_{t}(\vec{a}_{j})})
    %f(I,\{I_{i}\}_{i\in\mathcal{A}_{timeless}})\vec{y},\vec{x})
    \\
    \textrm{ where }f_{\pm}(I,\{I_{k}\},\vec{y},\vec{x})=\frac{proj_{t}(\vec{y}-\vec{x})\pm proj_{x}(\vec{y}-\vec{x})}{2}
    \\
    -\sum_{k\in\mathcal{A}_{timeless}}I_{k}(\frac{proj_{t}(\vec{a}_{k})\pm proj_{x}(\vec{a}_{k})-d_{n}(0,\vec{a}_{k})}{2})-\frac{I}{2}
\end{multline}
\normalsize
\end{theorem}

Next, I want to demonstrate the capacity of $K_{n}(\vec{x},\vec{y})$ to approximate $K_{l_{2}^{*}}$. For that we first show that a limit of the continuum propagators exists. For this we have Theorem~\ref{thm:k1}.

\begin{theorem}
\label{thm:k1}

Consider the expression in Equation~\ref{eq:difference} for 

\begin{equation}
    \label{eq:difference}
    |K_{p}^{cont}(\vec{x}_{2},\vec{x}_{1})-K_{q}^{cont}(\vec{x}_{2},\vec{x}_{1})|\le \frac{\alpha proj_{t}(\vec{x}_{2}-\vec{x}_{1})}{min(p,q)}
\end{equation}

This expression shows that the sequence of continuum propagators is Cauchy in the sup norm (so long as the maximum t to evaluate this inequality is constrained). Therefore, it has a convergent subsequence, and $K_{l_{2}^{*}}$ exists. 

\end{theorem}

It will become clear from Theorem~\ref{thm:k3} that $K_{l_{2}^{*}}$ not only exists, but is non-trivial. In the context of \cite{pesky}, we will conjecture that this limiting function is the KG propagator and not the physically significant Feynman propagator. We wish to motivate our work by relating it to physical quantities; i.e. rigorously defining the physically meaningful propagators. To this end, we instead obtain the following theorems for \textbf{Feynman propogators}. We define (as the Feynman propogator $K_{n}^{Feyn}(\vec{x}_{2},\vec{x}_{1}):X\rightarrow\mathbb{C}$) the same discrete sum as in Equation~\ref{eq:propogators} except we allow the $\rho_{d_{n}}(\gamma)$ to be the sum of difference sequence elements of both $\pm d_{n}(\vec{a}_{i})$ for $\vec{a}_{i}\in \mathcal{A}_{n}$. As we will see in Lemma~\ref{lemma:sneaky}, this will reveal some beautiful symmetries underlying the Feynman Propagator. 

When we obtain the continuum object $K_{n}^{Feyn,cont}$, there is another change in convention we must adopt to obtain Theorem~\ref{thm:k2} relating to normalization. The aim of the divisor used in Equation~\ref{eq:cont} was to normalize the propagator such that its inverse Fourier transform would have a fixed integral. This becomes problematic as in the proof of Theorem~\ref{thm:k2} we would need to normalize a non-normalizable function; we would for the solution a propagator which had deltas along $x^{2}+I^{2}=t^{2}$. Therefore we adopt the following convention for our divisor. We shall find each $K_{n}^{cont}$, defined up to some time-dependent normalization, has a well-defined inverse Fourier transform in mass by Theorem~\ref{thm:k3}. The normalization adopted will be to multiply the regular normalization by some power of t, effectively leaving unconstrained the integral but constraining its pointwise value. %$x=I=0$.

Let $H_{0}^{(2)}$ denote the zero-th Bessel function of the second kind, we know from \cite{Hong_Hao_2010} that it is the form of the relativistic bosonic propagator $\mathbb{E}[T\phi(\vec{x})\phi(\vec{y})]$, and we derive it rigorously within Theorem~\ref{thm:k2}.

%\note{need to fix the citation}

\begin{theorem}
\label{thm:k2}

Let $d=1$. For $\vec{x},\vec{y}\in X=\mathbb{Z}\times \mathbb{Z}$ we have

$$lim_{n\rightarrow\infty}K_{n}^{Feyn}(\vec{x},\vec{y})=CH^{(2)}_{0}(md_{l_{2}^{*}}(\vec{x}_{2},\vec{x}_{1}))=C'\mathcal{F}|_{I}^{m}(\frac{1}{\sqrt{t^{2}-x^{2}-I^{2}}})$$%}{d_{l_{2}^{*}}(\vec{x}_{2},\vec{x}_{1})^{2}}$$

where $\mathcal{F}|_{I}^{m}$ is a fourier transform C' can possibly have a phase and then by definition

%\note{The m might be wrong here}

$$K^{Feyn}_{d_{l_{2}}^{*}}(\vec{x},\vec{y})=CH_{0}^{(2)}(md_{l_{2}^{*}}(\vec{x}_{2},\vec{x}_{1}))=C\mathbb{E}[T\phi(\vec{x})\phi(\vec{y})]$$

%}(md_{l_{2}^{*}}(\vec{x}_{2},\vec{x}_{1}))}{d_{l_{2}^{*}}(\vec{x}_{2},\vec{x}_{1})^{2}}$$

\end{theorem}

%Finally 

Now that we have established theorems that motivate the utility of our lattice propagators, we move to the $K_{1}$ propagator. This propagator is simple to compute in many contexts and allows for the study of the first-order properties of propagators on an arbitrary surface.

\textbf{A note for physicists}

The nuance explained before the statement of Theorem~\ref{thm:k1} applies to the section below. These results do not approximate the time ordered Klein Gordon propogator $\mathbb{E}[\phi(\vec{x})\phi(\vec{y})]$\cite{pesky} and the realistic $\mathbb{E}[T\phi(\vec{x})\phi(\vec{y})]$. To replicate this as we have done in Theorem~\ref{thm:k1}, we would also need to include in $\rho_{d_{n}}(\gamma)$ negative length terms. This is akin to considering anti-particles and permitting the spontaneous alteration of a particle to an anti-particle during its travel.

\subsection{Statement of $K_{1}$ Theorems}

First, I find analytic formulas for the $K_{1}$ propogator for $X=\mathbb{Z}^{d}\times \mathbb{Z}$. We will find non-implicit expressions for $K_{1}$ for arbitrary d and show how we would obtain the continuum limit for only $d=2$.

\begin{theorem}
\label{thm:freeprop}
Let $\vec{x}_{1},\vec{x}_{2}\in X$ such that $proj_{t}(\vec{x}_{2}-\vec{x}_{1})\ge 0$. Furthermore, let $|*|_{l_{1}}$ denote the taxicab norm on X. Then for $\vec{x}=\vec{x}_{2}-\vec{x}_{1}-proj_{t}(\vec{x}_{2}-\vec{t})\hat{t}$. I have

\begin{multline}
    \label{eq:analL1}
    K_{1}(\vec{x}_{2},\vec{x}_{1})=\sum_{I}\sum_{I_{1}=|x_{1}|}^{\frac{I+|\vec{x}|_{l_{1}}}{2}}\sum_{I_{2}=|x_{2}|}^{\frac{I+|\vec{x}|_{l_{1}}}{2}-|x_{1}|}....\sum_{I_{d}=|x_{d}|}^{\frac{I+|\vec{x}|_{l_{1}}}{2}-\sum_{i=1}^{d-1}|x_{i}|}C(\{I_{i}\})
    \\
    \textrm{ where }C(\{I_{i}\})=\frac{f(I,|\vec{x}|_{l_{1}})proj_{t}(\vec{x}_{2}-\vec{x}_{1})!}{(\Pi_{i=1}^{d}(I_{i})!(I_{i}-|proj_{x_{i}}(\vec{x})|)!)(proj_{t}(\vec{x}_{2}-\vec{x}_{1})-I)!}e^{i(I-proj_{t}(\vec{y}-\vec{x}))}
    \\
    \textrm{ where }f(n,m)=\begin{pmatrix}1&&n-m\in 2\mathbb{Z}\\0&&\textrm{ otherwise }\end{pmatrix}
\end{multline}

Letting $d=2$ we have

\begin{equation}
    \label{eq:analL1continuum}
    K_{1}^{cont}(\vec{x}_{2},\vec{x}_{1})=C(t)\mathcal{F}(\begin{Bmatrix}t\\.5(I-|x|),.5(I+|x|),t-I\end{Bmatrix})
\end{equation}

where $\mathcal{F}$ denotes the fourier transform.

\end{theorem}

I now will calculate $K_{1}$ for $\sim$ equivalence classes such that $X^{cont}=\mathbb{T}^{d}$ and the Klein Bottle. Future work will be devoted to an in-depth treatment of all $\sim$ that lead to Riemann Surfaces in $d=2$. Theorem~\ref{thm:torus} details our results for the torus.

\begin{theorem}
\label{thm:torus}

Let $d\ge 1$ and let all $L_{i}\in \mathbb{N}$. Furthermore, let $\sim$ be defined on $\partial X$ such that $\vec{x}\sim\vec{y}$ if there is some $i$ st $proj_{x_{i}}(\vec{x})=\pm L_{i}=\mp proj_{x_{i}}(\vec{y})$ and for all $j\neq i, proj_{x_{j}}(\vec{x})=proj_{x_{j}}(\vec{y})$. For ease of expression, we will denote $K_{1}^{\mathbb{T}}$ as the $K_{1}$ for this $X$ and $K_{1}$ for the expression derived in Equation~\ref{eq:analL1}. Then we have Equation~\ref{eq:torus}.

\begin{equation}
    \label{eq:torus}
    K_{1}^{\mathbb{T}}(\vec{x}_{2}-\vec{x}_{1})=\sum_{\{\vec{x}'|\vec{x}'=\vec{x}_{2}+\sum 2m_{i}L_{i}e_{i},|\vec{x}'-\vec{x}_{1}-proj_{t}(\vec{x}'-\vec{x}_{1})\hat{t}|_{l_{1}}\le proj_{t}(\vec{x}'-\vec{x}_{1}),m_{i}\in \mathbb{Z}\}}(K_{1}(0,\vec{x}'-\vec{x}_{1}))
\end{equation}

\end{theorem}

We have the result for the Klein bottle in Theorem~\ref{thm:klien}. This surface is non-orientable, so technically not a Reiman surface.

\begin{theorem}
\label{thm:klien}
Let $d=2$ and let all $L_{i}\in \mathbb{N}$. We let $\sim$ be defined on $\partial X$ such that $\vec{x}\sim\vec{y}$ if such that $proj_{x_{1}}(\vec{x})=\pm L_{1}=\mp proj_{x_{1}}(\vec{y})$ and $proj_{x_{2}}(\vec{x})=proj_{x_{2}}(\vec{y})$  or $proj_{x_{1}}(\vec{x})=-proj_{x_{1}}(\vec{y})$ and $proj_{x_{2}}(\vec{x})=\pm L_{1}=\mp proj_{x_{2}}(\vec{y})$. For ease of expression, we will denote $K_{1}^{klien}$ as the $K_{1}$ for this X and $K_{1}$ for the expression derived in Equation~\ref{eq:analL1}. Then, we have
\small
\begin{multline}
    \label{eq:klien}
    K_{1}^{klien}(\vec{x}_{2}-\vec{x}_{1})=\sum_{\vec{x}'\in\mathcal{B}}(K_{1}(0,\vec{x}'-\vec{x}_{1}))
    \\
    \mathcal{B}=\{\vec{x}'|\vec{x}'=(\vec{x}_{2}+ 2m_{1}L_{1}\vec{e}_{1}+2(m_{2}-proj_{x_{2}}(\vec{x}_{2}))L_{2}\vec{e}_{2},|\vec{x}'-\vec{x}-proj_{t}(\vec{x}'-\vec{x})\hat{t}|_{l_{1}}\le proj_{t}(\vec{x}'-\vec{x}),m_{1}\in \mathbb{Z}\}
\end{multline}
\normalsize
\end{theorem}

Finally, I want to include one result for X in the context of tropical surfaces. Let $d=2,c\in \mathbb{N}$ and let $\mathfrak{p}=max(\{\vec{a}*\vec{x}|\vec{a}\in\{(1,1,-1),(1,-1,-1),(-1,1,-1),(-1,-1,-1)\}\})-c$. Then, $\mathfrak{p}|_{proj_{t}(x)>0}=-d_{1}(0,x)-c$. One may recognize that this surface is the tropical equivalent of de-Sitter space for $t>0$. There are tropical versions of each de-Sitter space; for now, we will also let $X$ be the zero set of $\mathfrak{p}=-d_{1}(0,\vec{x})-c$ for any $d\ge 2$.

\begin{theorem}
\label{thm:deSitter}

Let $X$ be the zero set of $\mathfrak{p}=-d_{1}(0,\vec{x})-c$ (at $d=2$ this is a tropical surface for $proj_{t}(\vec{x})>0$. Say $\vec{x},\vec{y}\in X$ such that $proj_{t}(\vec{y}-\vec{y})>0$ and $proj_{t}(\vec{x})>0)$. Then we have Equation~\ref{equation:deSitd}.

\begin{equation}
    \label{equation:deSitd}
    K_{1}(\vec{y},\vec{x})=\begin{pmatrix}proj_{t}(\vec{y}-\vec{x})\\|proj_{x_{1}}(\vec{y}-\vec{x})|,...,|proj_{x_{d}}(\vec{y}-\vec{x})|\end{pmatrix}
\end{equation}

We also obtain $K_{1}^{cont}(\vec{y},\vec{x})$. It immediately arises from \cite{continuouslatticepath}.

\begin{equation}
    \label{equation:deSitd2}
    K^{cont}_{1}(\vec{y},\vec{x})=\frac{\begin{Bmatrix}proj_{t}(\vec{y}-\vec{x})\\proj_{x_{1}}(\vec{y}-\vec{x}),...,proj_{x_{d}}(\vec{y}-\vec{x})\end{Bmatrix}}{\begin{Bmatrix}proj_{t}(\vec{y}-\vec{x})\\ \frac{proj_{t}(\vec{y}-\vec{x})}{d},...,\frac{proj_{t}(\vec{y}-\vec{x})}{d}\end{Bmatrix}}
\end{equation}

\end{theorem}

I now move on to proving these theorems.

\section{Proofs of the $K_{n}$ Theorems}

First, let us obtain the proof of Theorem~\ref{thm:k3}

\begin{proof}

We let $\vec{x},\vec{y}\in X=\mathbb{Z}\times \mathbb{Z}$ such that $proj_{t}(\vec{y}-\vec{x})>0$. Let $\mathcal{A}_{n}$ be as derived in Theorem~\ref{thm:axes}. For $\gamma\in \Gamma_{n}$, we know by Theorem~\ref{thm:generate} that $\gamma$'s difference sequence may be drawn from $\mathcal{A}_{n}$; let $I_{\vec{a}}$ be the number of elements of the difference sequence of $\gamma$ equal to $\vec{a}$. Then, we note that for any path the following properties hold:

\begin{itemize}

\item $\mathcal{I}:\sum_{a\in\mathcal{A}_{n}}I_{a}proj_{x}\vec{a}=proj_{x}(\vec{y}-\vec{x})$

\item $\mathcal{II}:\sum_{a\in\mathcal{A}_{n}}I_{a}proj_{t}\vec{a}=proj_{t}(\vec{y}-\vec{x})$

\item $\mathcal{III}:\rho_{n}(\gamma)=\sum_{a\in\mathcal{A}_{n}}I_{a}d_{n}(0,\vec{a})$

\end{itemize}

Let us denote $\rho_{n}(\gamma)$ as $I$ for the moment and attempt to group together paths of the same phase. For a fixed phase I, we have $\sum_{a\in\mathcal{A}_{n}}I_{a}$ different 'positions' we may place our vectors $\vec{a}$ from $\mathcal{A}_{n}$. We must place $I_{a_{1}}$ of these into the collection for $I_{a_{1}}$, $I_{a_{2}}$ into the collection for $I_{a_{2}}$, etc. Namely, once we find a distinct combination of $\{I_{\vec{a}}\}_{\vec{a}\in\mathcal{A}_{n}}$ which satisfy $\mathcal{I}\rightarrow\mathcal{III}$, we have $\begin{pmatrix}(\sum_{a\in\mathcal{A}_{n}}I_{a})!\\\Pi_{a\in\mathcal{A}_{n}}(I_{\vec{a}})!\end{pmatrix}$ different possible paths. Altogether, this means we obtain a propagator $K_{n}(\vec{y},\vec{x})$ takes on the form

\begin{equation}
	\label{eq:propStep1}
	K_{n}(\vec{y},\vec{x})=\sum_{I}(\sum_{\{I_{a}\}_{\vec{a}\in\mathcal{A}_{n}}\in\mathcal{I},\mathcal{II},\textrm{ and }\mathcal{III}}\begin{pmatrix}\sum_{\vec{a}\in\mathcal{A}_{n}}I_{a}\\\Pi_{a\in\mathcal{A}_{n}}(I_{\vec{a}})\end{pmatrix})e^{mI}
\end{equation}

%Now this expression will be peaked in contributions around the terms closest to $I_{\vec{a}}$ having almost equal terms over all $\vec{a}$. Or $(I_{a}-\frac{\sum_{a\in\mathcal{A}_{n}}I_{a}}{|\mathcal{A}_{n}|})^{2}$.

We now want to remove the implicit conditions $\mathcal{I}\rightarrow\mathcal{III}$. To do so, we recognize that $I_{a}\ge 0,proj_{t}(\vec{a})\in\mathbb{N},\textrm{ and }proj_{t}(\vec{y}-\vec{x})\in\mathbb{N}$. We order $\vec{a}\in\mathcal{A}_{n}$ such that $\vec{a}_{k+1}$ is that which corresponds to the $k$ in Theorem~\ref{thm:axes}, $\vec{a}_{\pm 0}=(\pm 1,1)$, and the order on the indeces is $-0\le +0\le k$ for $k\in \{0,...,|\mathcal{A}_{n}|\}$. We know that $\frac{|A_{n}|-1}{2}$ corresponds to $a_{k}=(0,1)$ because it is always in $\mathcal{A}_{n}$ and must be in the center of the set $|\mathcal{A}_{n}|$ by symmetry. First, to generate a path, we may allow $I_{a_{1}}$ be any positive number below $min(\lfloor\frac{proj_{t}(\vec{y}-\vec{x})}{proj_{t}(\vec{a}_{1})}\rfloor,\lfloor\frac{I}{d_{n}(0,\vec{a}_{1}})\rfloor)$ (as before we have written anything else, this variable is unconstrained). Then, the bounds on $I_{a_{j}}$ as defined in Eq~\ref{eq:final} follow for $\vec{a}_{j},j\in\{2,...,n\},j\neq\frac{|\mathcal{A}_{n}|-1}{2}$ as we just enforce condition $\mathcal{I}$ and $\mathcal{III}$ on the partial sums.  With this convention established, our conditions become the bounds of Eq~\ref{eq:final} and of Equation~\ref{eq:propStep2}

\begin{equation}
    \label{eq:propStep2}
    \begin{pmatrix}I_{\frac{|\mathcal{A}_{n}|-1}{2}}=I-\sum_{i\neq +0,i\neq \frac{|\mathcal{A}_{n}|-1}{2},i\neq -0}^{|\mathcal{A}_{n}|}I_{a_{i}}d_{n}(0,\vec{a}_{i})\\I_{+0}+I_{-0}=proj_{t}(\vec{y}-\vec{x})-\sum_{i\neq +0,i\neq \frac{|\mathcal{A}_{n}|-1}{2},i\neq -0}^{|\mathcal{A}_{n}|}I_{a_{i}}proj_{t}(\vec{a}_{i})-I_{\frac{|\mathcal{A}_{n}|-1}{2}})\\I_{+0}-I_{-0}=proj_{x}(\vec{y}-\vec{x})-\sum_{i\neq +0,i\neq \frac{|\mathcal{A}_{n}|-1}{2},i\neq -0}^{|\mathcal{A}_{n}|}I_{a_{i}}proj_{x}(\vec{a}_{i}))\end{pmatrix}
\end{equation}

where we satisfy Eq.~\ref{eq:propStep2} last after fixing each $I_{a_{i}}$ for $u\in \{2,...,n\}\setminus\{\frac{|\mathcal{A}_{n}|-1}{2}\}$. Solving the last two of these expressions, we have Equation~\ref{eq:propStep3}.

\begin{equation}
    \label{eq:propStep3}
    \begin{pmatrix}I_{\frac{n-2}{2}}=I-\sum_{i\neq +0,i\neq \frac{|\mathcal{A}_{n}|-1}{2},i\neq -0}^{|\mathcal{A}_{n}|}I_{a_{i}}d_{n}(0,\vec{a}_{i})\\I_{+0}=\frac{proj_{t}(\vec{y}-\vec{x})+proj_{x}(\vec{y}-\vec{x})}{2}-\sum_{i\neq +0,i\neq \frac{|\mathcal{A}_{n}|-1}{2},i\neq -0}^{|\mathcal{A}_{n}|}I_{a_{i}}\frac{proj_{t}(\vec{a}_{i})+proj_{x}(\vec{a}_{i})-d_{n}(0,\vec{a}_{i})}{2}-\frac{I}{2}\\I_{-0}=\frac{proj_{t}(\vec{y}-\vec{x})-proj_{x}(\vec{y}-\vec{x})}{2}-\sum_{i\neq +0,i\neq \frac{|\mathcal{A}_{n}|-1}{2},i\neq -0}^{|\mathcal{A}_{n}|}I_{a_{i}}\frac{proj_{t}(\vec{a}_{i})-proj_{x}(\vec{a}_{i})-d_{n}(0,\vec{a}_{i})}{2}-\frac{I}{2}\end{pmatrix}
\end{equation}

Giving us the final near desired expression of Eq~\ref{eq:final}.

\begin{multline}
	\label{eq:final}
	K_{n}(\vec{y},\vec{x})=\sum_{I}(\sum_{I_{a_{1}}=0}^{C_{1}}...\textrm{no k}=\frac{|\mathcal{A}_{n}|-1}{2}...\sum_{I_{a_{|\mathcal{A}_{n}|}}=0}^{C_{|\mathcal{A}_{n}|}}\begin{pmatrix}\sum_{\vec{a}\in\mathcal{A}_{n}}I_{a}\\\Pi_{a\in\mathcal{A}_{n}}(I_{\vec{a}})\end{pmatrix})e^{mI}
	\\
	C_{j}=min(\lfloor\frac{proj_{t}(\vec{y}-\vec{x})-\sum_{i=1,i\neq\frac{|\mathcal{A}_{n}|-1}{2}}^{j-1}I_{a_{i}}proj_{t}(\vec{a}_{i})}{proj_{t}(\vec{a}_{j})}\rfloor,\lfloor\frac{I-\sum_{i=1,i\neq\frac{|\mathcal{A}_{n}|-1}{2}}^{j-1}I_{a_{i}}d_{n}(0,\vec{a})}{d_{n}(0,\vec{a}_{j})}\rfloor)
	\\
	\textrm{ where Eq~\ref{eq:propStep3} holds for $I_{\pm 0}$ and $I_{\frac{|\mathcal{A}_{n}|-1}{2}}$}
\end{multline}

This expression is correct barring a caveat: we must enforce the condition that Eq~\ref{eq:propStep3} are non-negative integers. The first term is clearly an integer, and if we draw our $a_{i}$ as in Eq~\ref{eq:final}, it will be positive because of our upper bounds on each $I_{a_{i}},i\in \{1,...,n\}$. We note the other two terms are either both integers or both half integers. The expression $\sum_{i\neq +0,i\neq \frac{|\mathcal{A}_{n}|-1}{2},i\neq -0}^{|\mathcal{A}_{n}|}I_{a_{i}}\frac{proj_{t}(\vec{a}_{i})\pm proj_{x}(\vec{a}_{i})-d_{n}(0,\vec{a}_{i})}{2}$ is the same regardless of the $\pm$ by the symmetry of $|\mathcal{A}_{n}|$, so it will not change the parity. The only other difference in the terms is the first two parts, which are of the form $x_{1}\pm x_{2}$. This is the same number mod two, so either are both odd or even. Furthermore, let $proj_{x}(\vec{y}-\vec{x})>0$, then $I_{+0}\ge I_{-0}$, so this all depends on $2I_{-0}$ being positive and even. We will let that remain our only implicit condition. Now, let's look at what applying $T_{cont}^{p}$ does. Our $\mathcal{T}_{cont}^{p}$ commutes past our sums over $I_{a_{i}}$ and applies directly to the multinomial coefficient. This is because said sum just determines how many of each element in $\mathcal{A}_{n}$ is in our path, and then the multinomial coefficient contains the contributions from each different segment numbered path for the said combination of elements. Let $\vec{x},\vec{y}\in \mathbb{R}^{d}_{+}\times \mathbb{R}_{+}$. We obtain Equation~\ref{eq:final2}.

\small
\begin{multline}
	\label{eq:final2}
	\mathcal{T}_{cont}^{p}K_{n}([p\vec{y}],[p\vec{x}])=\sum_{I}(\sum_{I_{a_{1}}=0}^{C_{1}}...\textrm{no k}=\frac{|\mathcal{A}_{n}|-1}{2}...\sum_{I_{a_{|\mathcal{A}_{n}|}}=0}^{C_{|\mathcal{A}_{n}|}}\mathcal{T}_{cont}^{p}\begin{pmatrix}\sum_{\vec{a}\in\mathcal{A}_{n}}I_{a}\\\Pi_{a\in\mathcal{A}_{n}}(I_{\vec{a}})\end{pmatrix})e^{\frac{mI}{p}}
	\\
	C_{j}=min(\lfloor\frac{proj_{t}([p(\vec{y}-\vec{x})])-\sum_{i=1,i\neq\frac{|\mathcal{A}_{n}|-1}{2}}^{j-1}I_{a_{i}}proj_{t}(\vec{a}_{i})}{proj_{t}(\vec{a}_{j})}\rfloor,\lfloor\frac{I-\sum_{i=1,i\neq\frac{|\mathcal{A}_{n}|-1}{2}}^{j-1}I_{a_{i}}d_{n}(0,\vec{a})}{d_{n}(0,\vec{a}_{j})}\rfloor)
	\\
	\textrm{ where Eq~\ref{eq:propStep3} holds for $I_{\pm 0}$ and $I_{\frac{|\mathcal{A}_{n}|-1}{2}}$}
\end{multline}
\normalsize

From Equation~\ref{eq:final2} we can try to write down an expression for $K_{n}^{cont}(\vec{y},\vec{x})$.

\small
\begin{multline}
	\label{eq:dnprop}
	K_{n}^{cont}(\vec{y},\vec{x})=lim_{p\rightarrow\infty}\frac{\frac{1}{p^{|\mathcal{A}_{n}|}}\sum_{I}(\sum_{I_{a_{1}}=0}^{C_{1}}...\textrm{no k}=\frac{|\mathcal{A}_{n}|-1}{2}...\sum_{I_{a_{|\mathcal{A}_{n}|}}=0}^{C_{|\mathcal{A}_{n}|}}\mathcal{T}_{cont}^{p}\begin{pmatrix}\sum_{\vec{a}\in\mathcal{A}_{n}}I_{a}\\\Pi_{a\in\mathcal{A}_{n}}(I_{\vec{a}})\end{pmatrix})e^{\frac{mI}{p}}}{\frac{1}{p^{|\mathcal{A}_{n}|}}\mathcal{T}_{cont}^{p}(max_{\vec{x'},\vec{x'}\in X|_{m}}(\left|\Gamma_{n}^{\vec{x'},\vec{y'}}|_{m}\right|))}
	\\
	C_{j}=min(\lfloor\frac{proj_{t}([p(\vec{y}-\vec{x})])-\sum_{i=1,i\neq\frac{|\mathcal{A}_{n}|-1}{2}}^{j-1}I_{a_{i}}proj_{t}(\vec{a}_{i})}{proj_{t}(\vec{a}_{j})}\rfloor,\lfloor\frac{I-\sum_{i=1,i\neq\frac{|\mathcal{A}_{n}|-1}{2}}^{j-1}I_{a_{i}}d_{n}(0,\vec{a})}{d_{n}(0,\vec{a}_{j})}\rfloor)
	\\
	\textrm{ where Eq~\ref{eq:propStep3} holds for $I_{\pm 0}$ and $I_{\frac{|\mathcal{A}_{n}|-1}{2}}$}
\end{multline}
\normalsize

As is done in Theorem~\ref{thm:freeprop} for $d=2$, we absorb $p^{|\mathcal{A}_{n}|}$ into the numerator and denominator. And just as was done there, we note that the denominator converges to some integral over a continuous multinomial coefficient, which is finite by Theorem~\ref{thm:conttails}. That leaves us to consider the numerator.  Let $\mathcal{A}_{timeless}=\{k|\vec{a}_{k}\in\mathcal{A}_{n}\setminus\{(\pm 1,1),(0,1)\}\}$. In the theme of Theorem~\ref{thm:freeprop}, if the maximum of the multinomial in Equation~\ref{eq:dnprop} lies in the constraints of Equation~\ref{eq:dnprop} and if we have infinitely large magnitude I we sum over, then this converges to Equation~\ref{eq:FourN}, the well-defined function from $X^{cont}\rightarrow \mathbb{C}$

\tiny
\begin{multline}
    \label{eq:FourN}
    \mathcal{F}|_{I}^{m}(\int_{0}^{C_{i}}...\int_{0}^{C_{|\mathcal{A}_{timeless}|}}\begin{Bmatrix}\sum_{k\in\mathcal{A}_{timeless}}I_{k}+(I-\sum_{k\in\mathcal{A}_{timeless}}I_{k}d_{n}(0,\vec{a}_{k}))+f_{+}+f_{-}\\\{I_{k}\}_{k\in\mathcal{A}_{timeless}},I-\sum_{k\in\mathcal{A}_{timeless}}I_{k}d_{n}(0,\vec{a}_{k}),f_{+},f_{-}\end{Bmatrix}\Pi_{i\in\mathcal{A}_{timeless}}dI_{k})
    \\
    \textrm{ where }C_{j}=min(\frac{proj_{t}(\vec{y}-\vec{x})-\sum_{i=1,i\neq\frac{|\mathcal{A}_{n}|-1}{2}}^{j-1}I_{a_{i}}proj_{t}(\vec{a}_{i})}{proj_{t}(\vec{a}_{j})},\frac{proj_{t}(\vec{y}-\vec{x})-\sum_{i=1,i\neq\frac{|\mathcal{A}_{n}|-1}{2}}^{j-1}I_{a_{i}}d_{n}(0,\vec{a})}{proj_{t}(\vec{a}_{j})})
    %f(I,\{I_{i}\}_{i\in\mathcal{A}_{timeless}})\vec{y},\vec{x})
    \\
    \textrm{ where }f_{\pm}(I,\{I_{k}\},\vec{y},\vec{x})=\frac{proj_{t}(\vec{y}-\vec{x})\pm proj_{x}(\vec{y}-\vec{x})}{2}
    \\
    -\sum_{k\in\mathcal{A}_{timeless}}I_{k}(\frac{proj_{t}(\vec{a}_{k})\pm proj_{x}(\vec{a}_{k})-d_{n}(0,\vec{a}_{k})}{2})-\frac{I}{2}
\end{multline}
\normalsize

We note from our discrete sum over $I_{a_{1}}$ reaches arbitrarily high bounds. Since the continuous multinomial is rapidly decaying for small I and growing for large, this tells us we get the entire integration range for our integral. Similarly, we note the upper argument of the continuous multinomial can be expressed as $t-\sum_{k}(proj_{t}(\vec{a}_{k})-1)I_{k}$, which is bounded by $t$ and so our function is Schwartz class by Theorem~\ref{thm:conttails} and has a defined Fourier transform.

%So we need only show that this maxima is within our original bounds.

\end{proof}

Now that we have obtained $K_{n}^{cont}$, let us obtain a proof showing that $K_{n}^{cont}\rightarrow K_{l_{2}^{*}}^{cont}$, i.e. Theorem~\ref{thm:k1}

\begin{proof}

To begin, we want a function that will let us see paths in $\Gamma_{p}$ as approximate paths in $\Gamma_{q}$ for $q\ge p$. This approximation will let us constrain the difference between continuum propagators. Let's define the function $\pi_{N}:\mathcal{A}_{p+N}\rightarrow \mathcal{A}_{p}$ as follows. Let $\vec{a}\in \mathcal{A}_{p+N}$, and consider the set $argmin(\{|tan(\frac{proj_{t}(\vec{v})}{proj_{x}(\vec{v})}+\frac{\pi}{2})-tan(\frac{proj_{t}(\vec{a})}{proj_{x}(\vec{a})}+\frac{\pi}{2})|,\vec{v}\in\mathcal{A}_{p}\})$. This set has at most two vectors; if it contains one, we let $\pi_{N}(\vec{a})$ be said vector. If it has two, we let $\pi_{N}(\vec{a})$ be the vector of the two with least $proj_{x}(\vec{v})$. This gives us a well-defined function between our individual steps. Clearly, this extends to $\pi^{*}_{N}:\Gamma_{p+N}\rightarrow\Gamma_{p}$ by applying $\pi$ to the unique difference sequence of $\gamma$ in $\Gamma_{p+N}$ (given by Theorem~\ref{thm:generate}) to obtain a different sequence of a new path, denoted $\pi^{*}_{N}(\gamma)\in \Gamma_{p}$.

Let $\vec{u}_{p}^{1},\vec{u}_{p}^{2}$ be vectors aligned with $\vec{a}\in \mathcal{A}_{p+N}$ and $\pi_{p}(\vec{a})$, respectively, but such that $d_{l_{2}^{*}}(0,\vec{u}_{1})=1$. Then, as $p\rightarrow\infty$ it is necessarily the case that $d_{l_{2}}(\vec{u}_{2},\vec{u}_{1})=O(\frac{1}{p})$. This is because we know that there are asymptotically $\frac{p+N}{2\pi}$ circularly equidistributed points in $\mathcal{A}_{p+N}$ and $\frac{p}{2\pi}$ in $\mathcal{A}_{p}$ (by the circle equidistance Theorem~\ref{thm:tripledensity}), meaning that by the infinum definition of $\pi_{p}$, each vector in $\mathcal{A}_{q}$ is $O(\frac{1}{p})$ in $l_{2}$ and $l_{2}^{*}$ distance from its image under $\pi_{p}$ in $\mathcal{A}_{p}$ and that $|\pi_{p}^{-1}(\vec{a})|=\Theta(\frac{p+N}{p})$. Let $|\gamma|$ denote the number of difference sequence elements of $\gamma$. This second fact implies that $|\pi^{*-1}_{p}(\gamma)|=\frac{p+N}{p}^{|\gamma|}$ as the preimage of $\gamma\in \mathcal{A}_{p}$ includes for every difference sequence element a point in $\pi_{p}^{-1}$ of said difference sequence element. %(and there are t difference sequence elements).

%$$|K_{n}^{cont}(\vec{x}_{1},\vec{x}_{2})-K_{l_{2}^{*}}^{cont}(\vec{x}_{1},\vec{x}_{2})|$$

Let us define $G_{n}$ as in Equation~\ref{eq:denom}. We will define $t_{avg}$ in the next paragraph

%|\{\gamma\in \Gamma_{n}^{0,t\hat{t}},\rho_{n}=0\}|
\begin{equation}
    \label{eq:denom}
    G_{p}=(\frac{t}{t_{avg}})^{\frac{p}{4\pi}}max_{\vec{x}',\vec{y}'\in X,proj_{t}(\vec{y}'-\vec{x}')=t}(|\Gamma_{p}^{\vec{x}',\vec{y}'}|)
\end{equation}

Consider Theorem~\ref{thm:disctocont}. The continuous multinomial coefficient has a peak when all of its arguments are equal to their sum divided by $l$ and exponentially decays outside that range. Furthermore, when all the coefficients equal their sum divided by l, the Gaussian term becomes 1, and we are left with $\frac{l^{\sum I_{i}+\frac{l}{2}}}{\sqrt{2\pi \sum_{i=1}^{l}I_{i}}^{l-1}}$. By the work in Theorem~\ref{thm:k3} we know $\sum_{i=1}^{l}I_{i}$ equals $t-(\sum_{k}(proj_{t}(\vec{a}_{k})-1)I_{i})=t-(t_{avg}-1)(\sum_{i=1}^{l}I_{i})$ where $t_{avg}$ is the weighted average of the time increment in each walk. Because each $I_{i}$ will concentrate on being equally expressed, $t_{avg}$ is the average time increment of unit vectors along directions in $\mathcal{A}_{p}$; it is a constant function of $p$ indepednent of t and the linear segments in the directed paths the continuum multinomial expresses. By rearrangement this equals $\frac{t}{1+t_{avg}-1}=\frac{t}{t_{avg}}$. By substitution, this means the maximum of $|\Gamma_{n}|$ for time $t$ is Equation~\ref{eq:maximum}.

\begin{equation}
    \label{eq:maximum}
    \frac{l^{\frac{t}{t_{avg}}+\frac{l}{2}}}{\sqrt{2\pi \frac{t}{t_{avg}}}^{l-1}}
\end{equation}

Equation~\ref{eq:denom} we use in the denominator of $K^{cont}_{n}$ to normalize it, as $K_{n}(\vec{x}_{2},\vec{x}_{1})$ paths would grow on the order $\frac{proj_{t}(\vec{x}_{2},\vec{x}_{1})}{t_{avg}}^{\frac{n}{2\pi}}$ as there would be an average $t_{avg}$ time spent in line in the directed path and $\frac{n}{2\pi}$ directions for each line to take by the asymptotic property of pythagorean triples (Theorem~\ref{thm:tripledensity}). We've shown that the continuum limit under $\mathcal{T}^{cont}$ of $G_{l}$ grows as $(\frac{t}{t_{avg}})^{\frac{l}{2}}(\frac{t}{t_{avg}})(\frac{t}{1+\frac{\sum_{k=1}^{l}(proj_{t}(\vec{a}_{k})-1)I_{i}}{\sum I_{i}}})^{-\frac{l}{2}}(l^{\frac{t}{1+\frac{\sum_{k=1}^{l}(proj_{t}(\vec{a}_{k})-1)I_{i}}{\sum I_{i}}}})$ or $\Theta(\frac{t}{t_{avg}}^{\frac{l}{2}-\frac{l}{2}}l^{\frac{t}{t_{avg}}})$ for $t_{avg}$ an average time of each step in a directed path. This would imply that $\frac{G_{p}}{G_{q}}=\frac{t}{t_{avg}}^{\frac{q-p}{4\pi}}(\frac{p}{q})^{\frac{t}{t_{avg}}}$. Our same proof above allows us to conclude that $|\gamma|=\frac{t}{t_{avg}}$ and we note $|\pi^{*-1}_{p}(\gamma)|=\frac{q}{p}^{\frac{t}{t_{avg}}}$ so $\frac{G_{p}}{G_{q}}|\pi^{*-1}_{p}(\gamma)|=(\frac{p}{q})^{\frac{t}{t_{avg}}-\frac{t}{t_{avg}}}=1$. These last steps hold approximately around the sharp maximum of the continuum multinomial coefficient, where Stirling's Approximation may be used to find the subleading terms in $t$ and $p$.

%is the coefficient (probably e) that $l^{\frac{1}{1+\sum_{k=1}^{l}(proj_{t}(\vec{a}_{k})-1)}}$ approximates at large l.

%As a consequence of the multinomial theorem, when we sum over all entries in the discrete multinomial we obtain $l^{\sum_{i}x_{i}}$ where l is the number of arguments of the multinomial and $\sum_{i}x_{i}$ is the sum over those arguments. Integration and summing commute past $\mathcal{T}^{cont}$, and so when we integrate

%\note{STILL NEED WORK HERE}

These two tools allow us to approach our Cauchy claims. Consider first the terms in Equation~\ref{eq:propdiff}:

\begin{multline}
    \label{eq:propdiff}
    |\frac{K_{p}(\vec{x}_{1},\vec{x}_{1})}{G_{p}}-\frac{K_{q}(\vec{x}_{1},\vec{x}_{1})}{G_{q}}|=|\frac{\sum_{\gamma\in\Gamma_{p}^{\vec{x}_{2},\vec{x}_{1}}}e^{i\rho_{p}(\gamma)}}{G_{p}}-\frac{\sum_{\gamma\in\Gamma_{q}^{\vec{x}_{2},\vec{x}_{1}}}e^{i\rho_{q}(\gamma)}}{G_{q}}| 
    \\
    \le |\frac{\sum_{\gamma\in\Gamma_{p}^{\vec{x}_{2},\vec{x}_{1}}}e^{i\rho_{p}(\gamma)}}{G_{p}}-\frac{\sum_{\gamma\in\Gamma_{p}^{\vec{x}_{2},\vec{x}_{1}}}e^{i\rho_{l_{2}^{*}}(\gamma)}}{G_{p}}|+|\frac{\sum_{\gamma\in\Gamma_{q}^{\vec{x}_{2},\vec{x}_{1}}}e^{i\rho_{l_{2}^{*}}(\gamma)}}{G_{q}}-\frac{\sum_{\gamma\in\Gamma_{q}^{\vec{x}_{2},\vec{x}_{1}}}e^{i\rho_{q}(\gamma)}}{G_{q}}|\\
    +|\frac{\sum_{\gamma\in\Gamma_{p}^{\vec{x}_{2},\vec{x}_{1}}}e^{i\rho_{l_{2}^{*}}(\gamma)}}{G_{p}}-\frac{\sum_{\gamma\in\Gamma_{q}^{\vec{x}_{2},\vec{x}_{1}}}e^{i\rho_{l_{2}^{*}}(\gamma)}}{G_{q}}|
\end{multline}

We obtain Equation~\ref{eq:propdiff} by the Triangle Inequality. We note that some perimeters are equal ($\rho_{p}(\gamma_{1})=\sum_{i=0}^{n-1}d_{p}(x_{i+1},x_{i})=\sum_{i=0}^{n-1}d_{l_{2}^{*}}(x_{i+1},x_{i})=\rho_{d_{l_{2}^{*}}}(\gamma_{1})$ for $\{x_{i+1}-x_{i}\}_{i=1}^{n-1}$ a difference sequence of $\gamma_{1}$), so the first and second terms immediately disappear. Now, say $q=p+N$ where $N\in\mathbb{N}$, and $p$ is allowed to vary. Then, we can add and subtract a term to obtain Equation~\ref{eq:propdiffers3}.

\begin{multline}
    \label{eq:propdiffers3}
    |\frac{\sum_{\gamma\in\Gamma_{p}^{\vec{x}_{2},\vec{x}_{1}}}e^{i\rho_{l_{2}^{*}}(\gamma)}}{G_{p}}-\frac{\sum_{\gamma\in\Gamma_{q}^{\vec{x}_{2},\vec{x}_{1}}}e^{i\rho_{l_{2}^{*}}(\gamma)}}{G_{q}}|
    \\
    \le |\frac{\sum_{\gamma\in\Gamma_{p}^{\vec{x}_{2},\vec{x}_{1}}}e^{i\rho_{l_{2}^{*}}(\gamma)}}{G_{p}}-\frac{\sum_{\gamma\in\Gamma_{q}^{\vec{x}_{2},\vec{x}_{1}}}e^{i\rho_{l_{2}^{*}}(\pi_{p}^{*}(\gamma))}}{G_{q}}|+|\frac{\sum_{\gamma\in\Gamma_{q}^{\vec{x}_{2},\vec{x}_{1}}}(e^{i\rho_{l_{2}^{*}}(\pi_{p}^{*}(\gamma))}-e^{i\rho_{l_{2}^{*}}(\gamma)})}{G_{q}}|
\end{multline}

We can modify the first term of Equation~\ref{eq:propdiffers3} to obtain Equation~\ref{eq:propdiffers4}. To obtain this expression, we note that $\pi_{p}^{*}(\gamma)\in\Gamma_{p}$; so, while we originally had a sum over elements of $\Gamma_{q}$, we can rearrange our sum grouping together all $\gamma$ that are mapped to the same element $\pi_{p}^{*}(\gamma)$. This results in the cardinality of the preimage of a path $\gamma\in \Gamma_{p}$ being present in Equation~\ref{eq:propdiffers4}.
\tiny
\begin{multline}
    \label{eq:propdiffers4}
    |\frac{\sum_{\gamma\in\Gamma_{p}^{\vec{x}_{2},\vec{x}_{1}}}e^{i\rho_{l_{2}^{*}}(\gamma)}}{G_{p}}-\frac{\sum_{\gamma\in\Gamma_{q}^{\vec{x}_{2},\vec{x}_{1}}}e^{i\rho_{l_{2}^{*}}(\pi_{p}^{*}(\gamma))}}{G_{q}}|=|\frac{\sum_{\gamma\in\Gamma_{p}^{\vec{x}_{2},\vec{x}_{1}}}e^{i\rho_{l_{2}^{*}}(\gamma)}}{G_{p}}-\frac{\frac{G_{p}}{G_{q}}\sum_{\gamma\in\Gamma_{p}^{\vec{x}_{2},\vec{x}_{1}}}|(\pi_{p}^{*})^{-1}(\gamma)|e^{i\rho_{l_{2}^{*}}(\gamma)}}{G_{p}}|
    \\=|\frac{\sum_{\gamma\in\Gamma_{p}^{\vec{x}_{2},\vec{x}_{1}}}(1-\frac{G_{p}}{G_{q}}|(\pi_{p}^{*})^{-1}(\gamma)|)e^{i\rho_{l_{2}^{*}}(\gamma)}}{G_{p}}|
\end{multline}
\normalsize
If we demonstrate that if $\frac{G_{p}}{G_{q}}|(\pi_{p}^{*})^{-1}(\gamma)|\sim 1$ when we take $\mathcal{T}^{cont}$, this expression will disappear. This follows from our calculation of leading order approximations of $G_{n}$ and $|(\pi_{p}^{*})^{-1}(\gamma)|$ that we found above. Let us constrain the second expression of Equation~\ref{eq:propdiffers3}; we do so in Equation~\ref{eq:rsteppe}.

\begin{multline}
        \label{eq:rsteppe}
    G_{q}^{-1}|\sum_{\gamma\in\Gamma_{q}}e^{i\rho_{l_{2}^{*}}(\pi_{p}^{*}(\gamma))}-\sum_{\gamma\in\Gamma_{q}}e^{i\rho_{l_{2}^{*}}(\gamma)}|=G_{q}^{-1}\sum_{\gamma\in\Gamma_{q}}|e^{i\rho_{l_{2}^{*}}(\pi_{p}^{*}(\gamma))}-e^{i\rho_{l_{2}^{*}}(\gamma)}|
    \\=G_{q}^{-1}\sum_{\gamma\in\Gamma_{q}}\sqrt{2-2cos(\rho_{l_{2}^{*}}(\pi_{p}^{*}(\gamma))-\rho_{l_{2}^{*}}(\gamma))}
\end{multline}

Here we make the observation that $\sqrt{2-2cos(x)}\le x$ to obtain Equation~\ref{eq:rsteppe2}.

\begin{equation}
    \label{eq:rsteppe2}
    \le G_{q}^{-1}\sum_{\gamma\in\Gamma_{q}}|\rho_{l_{2}^{*}}(\pi_{p}^{*}(\gamma))-\rho_{l_{2}^{*}}(\gamma)|=G_{q}^{-1}(|\Gamma_{q}|)\frac{\alpha}{p}
    %=\alpha_{1}(pG_{p})^{-1}\sum_{\gamma\in\Gamma_{q}}\rho_{l_{2}^{*}}(\gamma)\le\frac{\alpha_{2}\tau}{p}
    %\le f(a)\sum_{\gamma\in\Gamma'}\rho_{l_{2}^{*}}(\gamma)
\end{equation}

where $\frac{\alpha}{p}$ is a uniform bound on the difference of $\rho_{l_{2}^{*}}(\gamma)$ and $\rho_{l_{2}}(\pi_{p}^{*}(\gamma))$. We know each vector is $O(\frac{1}{p})$ can be chosen to be uniformly separated from each other in $l_{2}$ norm by $\sim\frac{1}{p}$, and $l_{2}^{*}$ is absolutely continuous with $l_{2}$ (not visa versa as null-like vectors can have quite different $l_{2}$ and equal $l_{2}^{*}$). We bounded $\Gamma_{q}$ and $G_{q}$ in the work above; this demonstrates the result.

\end{proof}

We want to provide Lemma~\ref{lemma:sneaky}, whose importance to this proof allows it to be excluded from the definitional proofs of Section~\ref{section:def}. To use and even prove lemma~\ref{lemma:sneaky} effectively, we must relate it to the continuous multinomial of Equation~\ref{eq:FourN}. For this, we require a more tautological lemma in the form of Lemma~\ref{lemma:idea}. Let $\vec{x}\in \mathbb{R}^{3}\textrm{ and }p\in 2\mathbb{N}+2$ such that $\vec{x}=(x,I,t)$ and consider the paths to $\vec{x}$ composed of a difference sequence among $\mathcal{A}_{p}$ in Equation~\ref{eq:directions}.

\tiny
\begin{equation}
    \label{eq:directions}
    \mathcal{A}_{p}=\{(t_{q}e^{i\phi},t_{q})|\phi\in 2\pi(\frac{k}{p}),k\in\{0,...,p\},\begin{pmatrix}-t_{q}\\t_{q}\end{pmatrix}=\begin{pmatrix}-\frac{1}{p}&&1\\1&&-\frac{1}{p}\end{pmatrix}^{q}\begin{pmatrix}1\\1\end{pmatrix},\textrm{ and }k\in\{0,...,p\},q\in\{-p,...,p\}\}
    %\\
    %\cup\{(\pm 1,0)\}
\end{equation}
\normalsize

We note that $\mathcal{A}_{p}$ is defined so that its points are uniformly distributed  across the surface $t^{2}=x^{2}+I^{2}$ (in direction atleast; in magnitude we are incapable of doing so due to it being non-compact). Consider $D(n,p)$, the set of Smirnov words of length n and p many letters. We can associate to each letter $\{c_{k,q}\}_{k=0,q=0}^{p}$ in these Smirnov words a direction in $\mathcal{A}_{p}$. Let $\rho_{p}$ for Lemma~\ref{lemma:idea} denote the length obtained from $d_{p}$ on a path in $\Gamma_{p}$ (paths whose sets lie in $\mathcal{A}_{p}$) where $d_{p}$ is the metric whose axes of symmetry are $\mathcal{A}_{p}$

\begin{lemma}
    \label{lemma:idea}

    Consider the function $F_{p}$ in Equation~\ref{eq:rotafunc}
%max_{\vec{x}',\vec{y}'\in X|_{m},proj_{t}(\vec{y}'-\vec{x}')=t}(\sum_{n=1}^{\infty}\sum_{c\in D(n,p)}\mu(P(\vec{y}'-\vec{x}',c)))
\begin{equation}
    \label{eq:rotafunc}
    F_{p}(\vec{x})=\frac{\sum_{n=1}^{\infty}\sum_{c\in D(n,p)}\mu(P(\vec{x},c))}{(\frac{t}{t_{avg}})^{\frac{p}{4\pi}}max_{\vec{x}',\vec{y}'\in X,proj_{t}(\vec{y}'-\vec{x}')=t}(|\Gamma_{p}^{\vec{x}',\vec{y}'}|)}
\end{equation}

All portions of Equation~\ref{eq:rotafunc} are as they defined in \cite{continuouslatticepath}, referred to in Section~\ref{section:intro}, and rigorously shown to exist in Section~\ref{section:def}. Here, $D(n,p)$ represents a Smirnov word of length $n$ with p letters, and $P(\vec{x},c)$ is the polytope of directed paths from $0$ to $\vec{x}$ corresponding to that word with steps from $\mathcal{A}_{p}$. Then, if we let $\mathcal{A}_{timeless}$ denote the set of vectors in $\mathcal{A}_{p}$ except $(0,1)$, and if $d_{p}$ is the polygonal metric whose axes of symmetry is $\mathcal{A}_{p}$, then $F^{cont}_{p}$ is equal to the argument of the Fourier transform in Equation~\ref{eq:FourN} with the appropriate substitutions $\mathcal{A}_{timeless}$ and $d_{n}$ (with $d_{p}$) divided by the denominator of Equation~\ref{eq:rotafunc}.
    
\end{lemma}

\begin{proof}

We can show that the two desired expressions are equal by showing they are the same limit of $\mathcal{T}^{cont}_{m}$ of some discrete propagator-like object. Consider Equation~\ref{eq:tempdiscretespace}.

\begin{equation}
    \label{eq:tempdiscretespace}
    X^{temp}=\{\vec{a}\in\mathbb{R}^{3}|\vec{a}=\sum_{i=1}^{N}\vec{v}_{i},\{\vec{v}_{i}\}_{i=1}^{N}\subset \mathcal{A}_{p}\}\cup\{0\in\mathbb{R}^{3}\}
\end{equation}

These are all the points in $\mathbb{R}^{3}$ that can be approximated by some finite sum of elements in $\mathcal{A}_{p}$; we label it $X^{temp}$ to indicate it will only be used temporarily in the context of this proof. This space is isotropic like our original domain of $\mathbb{Z}^{3}$ in that any point may be treated as $0\in\mathbb{R}^{3}$ and you will still get the same points obtained from step-wise paths in $\mathcal{A}_{p}$ (just translated). Because for all $\vec{a}\in \mathcal{A}_{p}$ has $proj_{t}(\vec{a})=1$, the natural extension of the definition of $\Gamma^{\vec{x},\vec{y}}$ to this setting still obtains some finite collection of discrete paths. Similarly, we also would obtain a polygonal minkowski metric (call it $d_{p}$) from the same Equation~\ref{eq:polygon}, and from it, we could define $K_{p}(\vec{x},\vec{y})$. To define $K_{p}^{cont}$ in this setting, we need to define for $\vec{x}\in\mathbb{R}^{3}$ what the closest point function $[*]$ is. We may say:

$$[\vec{x}]=argmin(\{|\vec{x}-\vec{y}|_{l_{2}}|\vec{y}\in X^{temp}\})$$

This arg-min exists because the topology induced on $X^{temp}$ by $l_{2}$ is discrete; this fact is implied because each point in $X^{temp}$ must be achieved in finite time from some base point using steps st $proj_{t}(\vec{a})=1$. These are all the notions we need to define $K_{p}^{cont}$, and finally, if we replace $\mathcal{A}_{n}$ with $\mathcal{A}_{p}$ throughout all the steps of Theorem~\ref{thm:k3}, then we also would obtain as $K_{p}^{cont}$ the expression we find in the hypothesis of this lemma.
Now, we show $K_{p}^{cont}(\vec{x},\vec{y})$ has, as the argument to its Fourier transform, something that is equivalent to the polytopic sum of Lemma~\ref{lemma:sneaky}. In Theorem~\ref{thm:disctocont}, we see that all we need to show is that the discrete paths of $\Gamma^{\vec{x},\vec{y}}$ with some $l$ number of linear segments can be understood as discrete points lying in some continuum polytope (then because polytopes are Riemann integrable, the same arguments of Theorem~\ref{thm:disctocont} work in this setting). If we consider some Smirnov word composed of letters from $\mathcal{A}_{p}$ and length l, all of which describe one of the polytopes in the polytopic sum of Lemma~\ref{lemma:sneaky}, then we get that said polytope is a $l-3$ dimensional polytope embedded in $\mathbb{R}^{l}$. We see this in Figure~\ref{fig:latticedirec}.

\begin{figure}[H]
    \centering
    \includegraphics[width = 8 cm]{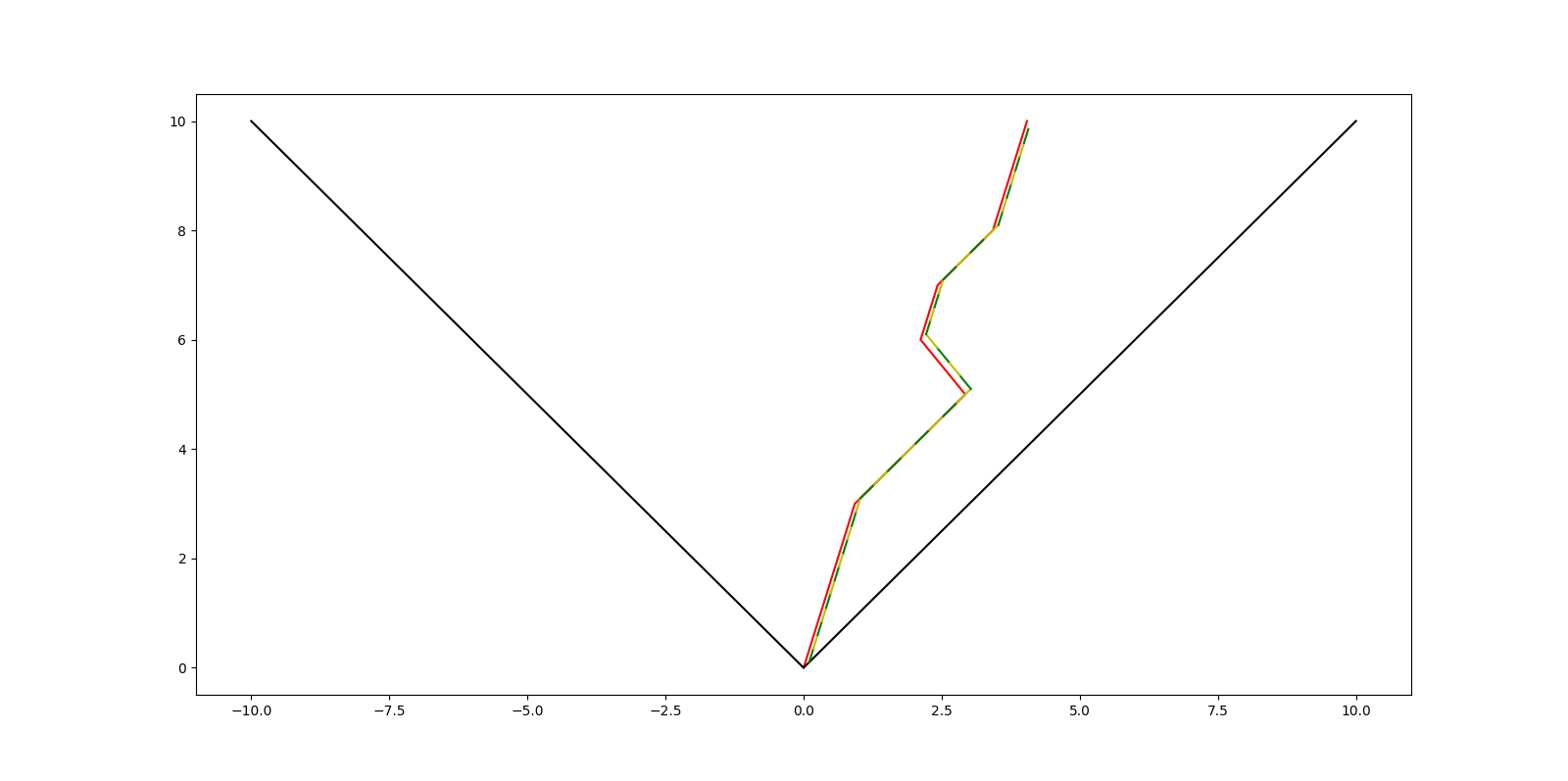}
    \caption{A lattice path (yellow and green) with integer coordinates approximating a red path from the path polytope}
    \label{fig:latticedirec}
\end{figure}

In Figure~\ref{fig:latticedirec}, we see an approximate lattice path to the directed one. It is clear from the figure that the integer coordinates of this lattice path would lie inside the polytope with l linear segments adjoining $\vec{x}_{1}$ to $\vec{x}_{2}$ (the polytope would include all integer paths adjoining $\vec{x}_{1}$ and $\vec{x}_{2}$ among its collection of direct paths). As we scale said polytope to larger sizes, it is still the case that the lattice paths between two points would be among the paths constituting the path polytope. The arguments of Theorem~\ref{thm:disctocont} proceed thereafter to obtain the desired result.

%this purpose we must expand upon the definition of the directed path polytope used in \cite{continuouslatticepath} and initially given in Section~\ref{section:intro}.

\end{proof}
%$\{c_{k,q}\}_{k=0,q=-p}^{p}$$H_{k}:D(n,2p^{2})\rightarrow D(n,2p^{2})$
For words  with letters among  $\{c_{k,q}\}_{k=0,q=-p}^{p}$, we have a mapping $H_{k}:D(n,2p^{2})\rightarrow D(n,2p^{2})$; this mapping takes $c_{j,q}$ to $c_{j+k,q}$ and is equivalent to rotating the underlying vectors by $2\pi\frac{k}{p}$ radians. If we let $R_{\theta}$ indicate the operator which rotates some $\vec{v}\in\mathbb{R}^{2}$ by $\theta$, then $H_{k}$ induces a volume preserving mapping from $P(q,c)$ and $P(R^{\frac{2\pi k}{p}}(proj_{x}(\vec{q})\hat{x}+proj_{I}(\vec{q})\hat{I})+proj_{t}(\vec{q})\hat{t},H_{k}(c))$.

%Another important, but (due to the non-compact nature of the unit hyperbola) more clumsy, mapping $I_{k}$ takes $c_{j,q}$ for $q-k\le p$ to $c_{j,q+k}$. If we let $B^{\eta}$ denote an operator which boosts some $\vec{x}\in \mathbb{R}^{2}$ by a velocity $\eta$, then $I_{k}$ induces an approximately volume preserving mapping from $P(re^{i\phi}+t\hat{t},c)$ and $P(proj_{r}(B^{\frac{1}{p}}(r,t))e^{i\phi}+proj_{t}(B^{\frac{1}{p}}(r,t))\hat{t},I_{k}(c))$.

From this we have Lemma~\ref{lemma:sneaky}.

\begin{lemma}{Sneaky Trick}
\label{lemma:sneaky}

Consider the function $F_{p}$ in Equation~\ref{eq:rotafunc}. Again, all portions of Equation~\ref{eq:rotafunc} are as they are defined in \cite{continuouslatticepath}, referred to in Section~\ref{section:intro}, and rigorously shown to exist in Section~\ref{section:def}. Here, $D(n,p)$ represents a Smirnov word of length $n$ with p letters, and $P(\vec{x},c)$ is the polytope of directed paths from $0$ to $\vec{x}$ corresponding to that word with steps from $\mathcal{A}_{p}$. Then, a sup norm cluster point of the sequence $\{F_{p}(\vec{x})\}_{p=1}^{\infty}$ exists and equals $(1-\frac{x^{2}+I^{2}}{t^{2}})^{-.5}$. 

\end{lemma}
\begin{proof}

First, let $\vec{x}'_{\theta}$ denote $R^{\theta}(proj_{x}(\vec{q})\hat{x}+proj_{I}(\vec{q})\hat{I})+proj_{t}(\vec{q})\hat{t}$. Then, we have Equation~\ref{eq:firstone}, where we denote the denominator of Equation~\ref{eq:rotafunc} as D.

%Furthermore let $D^{h}_{first}(n,p)\subset D(n,p)$ be Smirnov words with letters only among $\{c_{q}\}_{q=0}^{\frac{p}{2}-h}$ and $D^{h}_{last}(n,p)\subset D(n,p)$ be those with letters only among $\{c_{q}\}_{q=h}^{\frac{p}{2}}$ 

\begin{multline}
    \label{eq:firstone}
    F_{p}(\vec{x})=D^{-1}\sum_{n=1}^{\infty}\sum_{c\in D(n,p)}\mu(P(\vec{x},c))
    \\
    =D^{-1}\sum_{n=1}^{\infty}\sum_{c\in D(n,p)}\mu(P(\vec{x},H_{k}^{-1}(c)))
    \\
    =D^{-1}\sum_{n=1}^{\infty}\sum_{c\in D(n,p)}\mu(P(\vec{x}'_{\frac{2\pi k}{p}},c))=F_{p}(\vec{x}'_{\frac{2\pi k}{p}})
\end{multline}

We know $H_{k}$ is bijective on words $c\in D(n,p)$; hence, we can reorder terms in our sum between the first and second lines in our above equation depending on what $H_{k}^{-1}$ maps to. If we apply $H_{k}$ to every word throughout the whole sum, it will change into the same volume path polytope to $x'_{\frac{2\pi k}{p}}$ and revert $H_{k}^{-1}(c)$ giving us $F_{p}(\vec{x}'_{\frac{2\pi k}{p}})$. So, $F_{p}$ is invariant under the rotation group isomorphic to $\mathbb{Z}_{p}$, and any limit of it would be invariant under the rotation group on the first two coordinates and be a function of $x^{2}+I^{2}$. 

Let $x_{\eta}$ now denote $proj_{r}(B^{\eta}(r,t))e^{i\phi}+proj_{t}(B^{\frac{1}{p}}(r,t))\hat{t}$ and let $\oplus^{p} \eta$ denote the velocity obtained from iterating the einstein velocity summation formula (Eqaution~\ref{eq:einstein}) $p$ times using the same incremental velocity $\eta$.

\begin{equation}
    \label{eq:einstein}
    u\oplus v=\frac{u+v}{1+u*v}
\end{equation}

Then, we have Equation~\ref{eq:secone}, where again we denote the denominator of Equation~\ref{eq:rotafunc} as D.

\begin{multline}
    \label{eq:secone}
    |D^{-1}((DF_{p})(\vec{x})-(DF_{p})(\vec{x}^{\oplus^{k}\frac{1}{p}}))|
    \\=|D^{-1}\sum_{n=1}^{\infty}\sum_{c\in D(n,p)}(\mu(P(\vec{x},c))-\mu(P(\vec{x}^{\oplus^{k}\frac{1}{p}},c)))|
    \\    =D^{-1}|\sum_{n=1}^{\infty}\sum_{c_{1}\in D_{1}(n,p),c_{2}\in D_{2}(n,p)}(\mu(P(\vec{x},c_{1}))-\mu(P(\vec{x}^{\oplus^{k}\frac{1}{p}},c_{2})))|
    %\\\le p^{-1}
\end{multline}

In the last line of Equation~\ref{eq:secone} we create the sets $D_{1}(n,p)$ and $D_{2}(n,p)$. These are the subset of Smirnov words from our original alphabet of $\mathcal{A}_{p}$ that must include the last $k$ directions for $D_{1}$ nor the first for $D_{2}$. 

We note that $H_{k}$ naturally induces a volume-preserving association between $P(q,c_{i,j})$ and $P(q^{\oplus^{k}\frac{1}{p}},c_{i+k,q}))$, and we have subtracted those terms away to obtain the last line of Equation~\ref{eq:secone}. $H_{k}$ maps those directions outside of $D_{1}$ to directions wholely outside of $\mathcal{A}_{p}$; it must have a restricted domain to be a proper map between two sets of words. It is for this reason that the only terms left in the sum in Equation~\ref{eq:secone} lie outside the domain and image of $\mathcal{H}_{k}$ (and it is in this manner $D_{1}$ and $D_{2}$ obtain their definitions).

This last line has the upper bound $\frac{\sum_{n=1}^{\infty}V_{n}(|D_{1}(n,p)|+|D_{2}(n,p)|)}{D}$. The $V_{n}$ denotes the average volume of the path polytope over the Smirnov words in $D_{1}$ and $D_{2}$. The numerator of this fraction is less than the denominator (by its sup definition). As $p$ increases (with $k\sim \alpha p$) the fraction will approach zero. This is because $D$ is proportional to the total number of Smirnov words without the extra constraint that they have some letter in their composition, and this extra constraint means that $D$ is $\sim p-k$ times as large as $D_{1}$ or $D_{2}$. Looking at a word in $D_{1}$ and $D_{2}$, we can find the letter it is forced to have. By substituting this letter with the $p-k$ other directions in $\mathcal{A}_{p}$ we find there are $p-k$ more words in the composition of $D$ than $D_{1}$. An important thing to note is that as $p$ increases the $V_{n}$ remains invariant w.r.t p.

Since $x^{\oplus^{k}(\frac{1}{p})}$ constitutes a macroscopic boost, we have demonstrated that $DF_{p}$ at large $p$ is approximately boost invariant, not $F_{p}$. Since $D$ only depends on time, however, we have shown that a limit of $F_{p}$ has its dependence on x enveloped as a dependence on $\tau=\sqrt{t^{2}-x^{2}}$ and therefore it is a function of $\tau^{2}-I^{2}$. Namely, the limit over p of $F_{p}$ must be of the form $g(t)f(\tau^{2}-I^{2})$, where $g(t)\sim D$ is constrained by our normalization factor. 
%there is some function $\epsilon(p)>0$ such that $\epsilon(p)\rightarrow 0$
%$\etwwa\in (\frac{arctanh(\frac{-\tau+\epsilon(p)}{\taus})}{\tau},\frac{arctanh(\frac{\tau-\epsilon(p)}{\tau})}{\tau})$. 

Say that the variable $\theta=arcsin(\frac{I}{\tau})$ has a distribution as governed by $F_{p}$ that tends towards uniformity in the interval $[0,2\pi)$. Then we note (from $d\theta=\frac{dI}{\sqrt{\tau^{2}-I^{2}}}$) that $F_{p}$ is distributed in the desired manner on the domain $[-\tau,\tau]$ (from the Radon-Nykodym derivative of measures on $\mathbb{R}$).  The quantity $\theta$ is the angle of paths drawn from $\mathcal{F}_{p}$; due to the p-fold rotational symmetry of $\mathcal{F}_{p}$ we obtain the desired uniformity at $p\rightarrow\infty$. We have the desired result.

\end{proof}

Now, the main motivating theorem we will provide in this document is Theorem~\ref{thm:k2}. We prove this now:

\begin{proof}

First, we derive $K_{n}^{Feyn}(\vec{y},\vec{x})$. Note that we are just allowing in Theorem~\ref{thm:k3} for the vectors in $\mathcal{A}_{n}$ to have negative phases. Let $\mathcal{A}_{timeless}'$ denote the set of pairs $\{(\vec{a}_{n},\pm)|\vec{a}_{n}\in\mathcal{A}_{n}\}\setminus \{((0,1),+)\}$ where $\pm$ denotes whether we associate to $\vec{a}_{n}$ the length $\pm d_{n}(0,\vec{a}_{n})$. Then, $I_{k}$ for $k$ indexing $\mathcal{A}_{timeless}$ is the number of steps in $\gamma$ devoted to the pair $(\vec{a},\pm)_{k}$. We are purposely re-indexing our array here because our expression in Equation~\ref{eq:Feyn} becomes unwieldy otherwise. All the same proof of Theorem~\ref{thm:k3} shows that $K_{n}^{Feyn}$ is equal to the expression in Equation~\ref{eq:Feyn}.
\tiny
\begin{multline}
    \label{eq:Feyn}
   \mathcal{F}|_{I}^{m}(\int_{0}^{C_{i}}...\int_{0}^{C_{|\mathcal{A}_{timeless}|}}\begin{Bmatrix}\sum_{k\in\mathcal{A}_{timeless}'}I_{k}+(I-\sum_{k\in\mathcal{A}_{timeless}'}I_{k}d_{n}(0,\vec{a}_{k}))+f_{+}+f_{-}\\\{I_{k}\}_{k\in\mathcal{A}_{timeless}'},I-\sum_{k\in\mathcal{A}_{timeless}'}I_{k}d_{n}(0,\vec{a}_{k}),f_{+},f_{-}\end{Bmatrix}\Pi_{i\in\mathcal{A}_{timeless}'}dI_{k})
    \\
    \textrm{ where }C_{j}=min(\frac{proj_{t}(\vec{y}-\vec{x})-\sum_{\in\mathcal{A}_{timeless}'}^{j-1}I_{a_{i}}proj_{t}(\vec{a}_{i})}{proj_{t}(\vec{a}_{j})},\frac{proj_{t}(\vec{y}-\vec{x})-\sum_{\in\mathcal{A}_{timeless}'}^{j-1}I_{a_{i}}d_{n}(0,\vec{a})}{proj_{t}(\vec{a}_{j})})
    %f(I,\{I_{i}\}_{i\in\mathcal{A}_{timeless}})\vec{y},\vec{x})
    \\
    \textrm{ where }f_{\pm}(I,\{I_{k}\},\vec{y},\vec{x})=\frac{proj_{t}(\vec{y}-\vec{x})\pm proj_{x}(\vec{y}-\vec{x})}{2}
    \\
    -\sum_{k\in\mathcal{A}_{timeless}'}I_{k}(\frac{proj_{t}(\vec{a}_{k})\pm proj_{x}(\vec{a}_{k})-d_{n}(0,\vec{a}_{k})}{2})-\frac{I}{2}
\end{multline}
\normalsize

Note that the expression for Equation~\ref{eq:Feyn} is almost identical to Equation~\ref{eq:FourN} with the exception of the inclusion of $\mathcal{A}'_{timeless}$ indices. As a sanity check, $C_{j}$ is finite and defines a compact domain of integration and the denominator of our continuous multinomial less than $proj_{t}(\vec{y}-\vec{x})$ as expected. A barely modified proof as to what was employed in Theorem~\ref{thm:k2} serves to demonstrate after this that the pointwise limit of $K_{n}^{Feyn}$ exists and is non-trivial. We now want to use Lemma~\ref{lemma:sneaky} to complete the derivation of $K_{l_{2}^{*}}^{Feyn}(\vec{y},\vec{x})$, namely that the two sequences $\mathcal{F}^{-1}|_{m}^{I}(K_{n}^{Feyn})$ and that used in the statement of Lemma~\ref{lemma:sneaky} converge to the same limit.

For this, we note that $\mathcal{A}_{n}$ and the directions in the polytope of Lemma~\ref{lemma:sneaky} for $p=\frac{n}{2\pi}$ are both equidistributed on the sphere by Theorem~\ref{thm:tripledensity}. The expression in Equation~\ref{eq:FourN} is continuous w.r.t the directions in $\mathcal{A}_{n}$ used to construct it; therefore, if we took a sequence of the difference point-by-point of the argument of the Fourier transform of Equation~\ref{eq:FourN} for the original $\mathcal{A}_{n}$ and those directions in Equation~\ref{eq:directions} for $p=\frac{n}{2\pi}$, we obtain a sequence which goes to zero (by equidistributedness of pythagorean triples these directions ultimately approximate each-other). 

There are a couple of steps before we can employ Lemma~\ref{lemma:sneaky}. First, we note that our above work, along with Theorem~\ref{thm:k1}, demonstrates that Equation~\ref{eq:FourN} has a sup norm limit for the directions in Equation~\ref{eq:directions}. We now only need to show that the argument of Equation~\ref{eq:FourN} is equivalent to the expression in the hypothesis of Lemma~\ref{lemma:sneaky}. This last proof is performed by Lemma~\ref{lemma:idea}.

\end{proof}

\section{Proofs of the $K_{1}$ Theorems}

Let us state the proof of Theorem~\ref{thm:freeprop}:

\begin{proof}

Let $\gamma=\{\vec{x}_{i}\}_{i=1}^{n}\in \Gamma_{1}^{\vec{x},\vec{y}}$. Then, we can take $\vec{x}_{i}-\vec{x}_{i-1}=\{(\pm e_{1},1),...,(\pm e_{d},1),(0,1)\}$ where $e_{i}$ is a unit direction in $\mathbb{Z}^{d}$. This is because by Theorem~\ref{thm:generate} any path in $\Gamma_{1}$ can be generated by the elements of $\mathcal{A}_{1}$ and by Theorem\ref{thm:axes} $\mathcal{A}_{1}$ is the set above. We will denote by $I_{\eta i}$ the number of elements in the difference sequence $\{\vec{x}_{i+1}-\vec{x}_{i}\}_{i=1}^{n-1}$ are $(\eta e_{i},1)$ for $\eta\in\{-1,1\}$. Let's also denote by $I_{0}$ the number of elements that are $(0,1)$ Then, we must have the following for $\gamma$ to end at $\vec{y}$:

\begin{itemize}
    \item $I_{i}-I_{-i}=proj_{x_{i}}(\vec{y}-\vec{x})$
    
    \item $I_{0}+\sum_{i=1}^{d}I_{\pm i}=proj_{t}(\vec{y}-\vec{x})$
    
\end{itemize}

Since $d_{1}(0,e_{i})=0$, the phase is only determined by $I_{0}$. In other words, $\rho_{1}(\gamma)=I_{0}=proj_{t}(\vec{y}-\vec{x})-\sum_{i=1}^{d}I_{\pm i}$. We will let $I=proj_{t}(\vec{y}-\vec{x})-I_{0}$ so that $I-proj_{t}(\vec{y}-\vec{x})$ is our phase. From the second condition on our $I_{i}$ we know that the number of different sequences must be one less than the time coordinate because each element of $\mathcal{A}_{p}$ increments our path by 1. So $n-1=proj_{t}(\vec{y}-\vec{x})-1$. The number of paths with a fixed phase $I_{0}$ and these properties is a simple multinomial; the number of ways to sort $n-1$ points into our $I_{\pm i}$ and $I_{0}$. So we have Equation~\ref{eq:step1}

\small
\begin{multline}
    \label{eq:step1}
    K_{1}(\vec{x},\vec{y})=
    \\
    \sum_{I=|x|_{l_{1}}}^{t-1}(\sum_{I=\sum_{i=1}^{d}(I_{i}+I_{-i}),proj_{x_{i}}(\vec{x}-\vec{y})=I_{i}-I_{-i}}\frac{(proj_{t}(\vec{y}-\vec{x})-1)!}{\Pi_{i=1}^{d}(I_{i}!I_{-i}!)(proj_{t}(\vec{y}-\vec{x})-1-I)!})e^{im(I-proj_{t}(\vec{y}-\vec{x}))}
\end{multline}
\normalsize

In deriving our multinomial in Equation~\ref{eq:step1}, we must also recognize that possibly different $I_{i}$ can satisfy our initial constraints and include a sum over them. Using our second condition (and the fact that $I_{i}\ge proj_{x_{i}}(\vec{y}-\vec{x})$ this becomes

\begin{multline}
    \label{eq:step2}
    \sum_{I=|x|_{l_{1}}}^{t-1}(\sum_{\mathcal{I},\mathcal{II}}\frac{(proj_{t}(\vec{y}-\vec{x})-1)!}{\Pi_{i=1}^{d}(I_{i}!(I_{i}-proj_{x_{i}}(\vec{y}-\vec{x}))!)(proj_{t}(\vec{y}-\vec{x})-1-I)!})e^{im(I-proj_{t}(\vec{y}-\vec{x}))}
    \\\textrm{ with conditions }\mathcal{I}=\{I_{i}|I=\sum_{i=1}^{d}(2I_{i}-proj_{x_{i}}(\vec{y}-\vec{x}))\},\mathcal{II}=\{I_{I}|I_{i}\ge proj_{x_{i}}(\vec{y}-\vec{x})\}
\end{multline}

We note that condition $\mathcal{I}$ of Equation~\ref{eq:step2} can be rewritten $\sum_{i}I_{e_{i}}=\frac{I+|\vec{y}-\vec{x}-proj_{t}(\vec{x}-\vec{y})\hat{t}|_{l_{1}}}{2}$. For this, however, we will obtain non-integer values should $f(I,|\vec{y}-\vec{x}-proj_{t}(\vec{x}-\vec{y})\hat{t}|_{l_{1}})=0$. This corresponds to no path being possible between these two points, hence the inclusion of $f$ in our final expression. The condition on the sum (and each $I_{i}\ge 0$) leads to the expansion and our result; we sum over all allowed $I_{1}$ first, and then with a choice of $I_{1}$ all allowed $I_{2}$, etc. This will yield every combination consistent with our equations and the expression desired for Theorem~\ref{thm:freeprop}. Now, let $d=2$. Then, for $t\in \mathbb{R}^{+},x\in[-t,t]$ we have

\begin{equation}
    \label{eq:step3}
    \frac{\mathcal{T}_{cont}^{n}K_{1}(0,\lfloor nx\rfloor\hat{x}+\lfloor nt\rfloor\hat{t})}{\mathcal{T}_{cont}^{n}max_{\vec{x
'}\in\mathbb{Z}\times\mathbb{Z}}(|\Gamma_{1}^{0,\vec{x}'}|)}=\sum_{I=\lfloor nx\rfloor}^{\lfloor nt\rfloor}\frac{\mathcal{T}_{cont}^{n}\frac{\lfloor nt\rfloor!}{(.5(I-\lfloor nx\rfloor))!(.5(I+|\lfloor nx\rfloor|))!(\lfloor nt\rfloor-I)!}}{\mathcal{T}_{cont}^{n}max_{\vec{x
'}\in\mathbb{Z}\times\mathbb{Z}}(|\Gamma_{1}^{0,\vec{x}'}|)}e^{im(I-\lfloor nt\rfloor)}
\end{equation}

Using Sterling's formula we obtain Equation~\ref{equation:Stirling}.

\begin{multline}
    \label{equation:Stirling}
    \\
    log(\frac{\lfloor nt\rfloor!}{(.5(I-\lfloor nx\rfloor))!(.5(I+|\lfloor nx\rfloor|))!(\lfloor nt\rfloor-I)!})=
    \\
    \lfloor nt\rfloor log(\lfloor nt\rfloor)-.5(I-|\lfloor nx\rfloor|)log(.5(I-|\lfloor nx\rfloor|))-.5(I+|\lfloor nx\rfloor|)log(.5(I+|\lfloor nx\rfloor|))
    \\-(\lfloor nt\rfloor-I)log(\lfloor nt\rfloor-I)+ln(\frac{1}{\pi}\sqrt{ \frac{\lfloor nt\rfloor}{(I^{2}-|\lfloor nx\rfloor|^{2})(\lfloor nt\rfloor-I)}})
    \\
    +\mathcal{O}(\lfloor nt\rfloor^{-1}+2(I-|\lfloor nx\rfloor|)^{-1}+2(I+|\lfloor nx\rfloor|)^{-1}-(\lfloor nt\rfloor-I)^{-1})
\end{multline}

Upon taking a derivative, we obtain Equation~\ref{eq:deriv}.

\begin{equation}
    \label{eq:deriv}
    \partial_{I}log(\frac{\lfloor nt\rfloor!}{(.5(I-\lfloor nx\rfloor))!(.5(I+|\lfloor nx\rfloor|))!(\lfloor nt\rfloor-I)!})= log(\frac{2(t-I)}{\sqrt{I^{2}-x^{2}}})+\mathcal{O}(\frac{1}{n})
\end{equation}

We note up to first order, this expression equals $log(\frac{2(t-I)}{\sqrt{I^{2}-x^{2}}})$, which has the zero located at $I_{max}$ in Equation~\ref{eq:max}

\begin{equation}
    \label{eq:max}
    I_{max}=\frac{4\lfloor nt\rfloor-\sqrt{4\lfloor nt\rfloor^{2}-3\lfloor nx\rfloor^{2}}}{3}
\end{equation}

Using the continuous multinomial distribution concentration property found in Theorem~\ref{thm:conttails}, we know that $\frac{\lfloor nt\rfloor!}{(.5(I-\lfloor nx\rfloor))!(.5(I+|\lfloor nx\rfloor|))!(\lfloor nt\rfloor-I)!}$ concentrates about $I_{max}$ with dramatically vanishing terms outside of $\sim\sqrt{n}$. Our phase converges to the same constant value over these non-negligible terms because we divide by $n$ in its argument. This implies that

\begin{multline}
    \label{eq:step4}
    lim_{n\rightarrow\infty}\frac{\mathcal{T}_{cont}^{n}K_{1}(0,[ nx]\hat{x}+[ nt ]\hat{t})}{\mathcal{T}_{cont}^{n}max_{\vec{x
'}\in\mathbb{Z}\times\mathbb{Z}}(|\Gamma_{1}^{0,[n\vec{x}']}|)}
\\
=lim_{n\rightarrow\infty}\frac{\sum_{I=I_{max}-[n\sqrt{t}]}^{I_{max}+[n\sqrt{t}]}\mathcal{T}_{cont}^{n}\frac{\lfloor nt\rfloor!}{(.5(I-\lfloor nx\rfloor))!(.5(I+|\lfloor nx\rfloor|))!(\lfloor nt\rfloor-I)!}e^{im(\frac{I-\lfloor nt\rfloor}{n})}}{\mathcal{T}_{cont}^{n}max_{\vec{x
'}\in\mathbb{Z}\times\mathbb{Z}}(|\Gamma_{1}^{0,\vec{x}'}|)}
\\
=lim_{n\rightarrow\infty}\frac{\sum_{I=I_{max}-[n\sqrt{t}]}^{I_{max}+[n\sqrt{t}]}\begin{Bmatrix}t\\.5(I-|x|),.5(I+|x|),t-I\end{Bmatrix}e^{im(\frac{I-\lfloor nt\rfloor}{n})}}{\mathcal{T}_{cont}^{n}max_{\vec{x
'}\in\mathbb{Z}\times\mathbb{Z}}(|\Gamma_{1}^{0,\vec{x}'}|)}
\end{multline}

We can multiply the numerator and denominator (which will behave similarly sans phase) by n, and then this sum (by the definition of a Riemann integral) converges to some normalized Fourier transform of the continuous multinomial in terms of the variable I to $m$ in Equation~\ref{eq:finalEq}. This is because the bounds of our sum contain the sharp maxima of the continuous multinomial at $\frac{I_{max}}{n}$ between $x$ and $t$; therefore, this integral amounts to all of the integral required for the Fourier transform.

Note that in the denominator the factor of $\frac{1}{n}$ acts on a sum very much like the numerator but without a phase. This will obtain an integral of the continuous multinomial, but by Theorem~\ref{thm:conttails}, we know this is integrable, and we can absorb it into some $C(t)$.

\begin{equation}
    \label{eq:finalEq}
    C(t)\mathcal{F}_{I}^{m}(\begin{Bmatrix}t\\.5(I-|x|),.5(I+|x|),t-I\end{Bmatrix})
\end{equation}

By Theorem~\ref{thm:conttails}, we know this Fourier transform exists. This is worked out in detail. Similarly, the discrete propagator for all the other cases becomes Fourier transforms, and as the above work shows, we need only show that the peak of the continuous multinomial lies in our sum bounds (more rigorously that is within $\sqrt{n}t$ from the boundary of the continuous region). That would give us Equation~\ref{eq:finaBo2}.
\small
\begin{multline}
    \label{eq:finaBo2}
    C(t)\mathcal{F}_{I}^{m}(\int_{B_{1}}^{T_{i}}\textrm{ ... }\int_{B_{d-1}}^{T_{d-1}}\begin{Bmatrix}
    proj_{t}(\vec{x}_{2}-\vec{x}_{1})\\\{I_{d}\}_{i=1}^{d}\{I_{i}-|proj_{x_{i}}(\vec{x}_{2}-\vec{x}_{1})|\}_{i=1}^{d},proj_{t}(\vec{x}_{2}-\vec{x}_{1})-I
    \end{Bmatrix}\Pi_{i=1}^{d}dI_{i})
    \\\textrm{ where }B_{i}=|proj_{x_{i}}(\vec{x}_{2}-\vec{x}_{1})|\textrm{ and }T_{i}=\frac{I+|\vec{x}_{2}-\vec{x}_{1}|_{l_{1}}}{2}-\sum_{j=1}^{i-1}|proj_{x_{j}}(\vec{x}_{2}-\vec{x}_{1})|
\end{multline}
\normalsize
%C(t)\mathcal{F}_{I}^{m}(\begin{Bmatrix}t\\.5(I-|x|),.5(I+|x|),t-I\end{Bmatrix})

%Therefore we obtain the desired result. All the same work obtains 

%Since $proj_{t}(\vec{a})=1$ for all $\vec{a}\in \mathcal{A}_{p}$, we know that 

\end{proof}

With $K_{1}$ for $X=\mathbb{Z}^{d}\times \mathbb{Z}$ obtained in Theorem~\ref{thm:freeprop}, we may use it to obtain any Riemann Surface, which has a simple universal cover in $\mathbb{R}^{2}$. We will perform the proofs of Theorem~\ref{thm:torus} and Theorem~\ref{thm:klien} simultaneously using their well-known covering by $\mathbb{R}^{d}$ in the following paragraphs.

\begin{proof}

The covering $\phi:\mathbb{Z}^{d}\rightarrow \mathbb{T}^{d}$ takes $\vec{x}$ to its representation $\vec{x}'\in \times_{i}\{-L_{i},...,L_{i}\}$ such that there is some $\{m_{i}\}_{i=1}^{d}\subset \mathbb{Z}$ where $\vec{x}'=\vec{x}+m_{i}(2L_{i}e_{i})$. If we want to describe some $\gamma\in \Gamma_{1}$ for $\mathbb{T}^{d}$, it's sufficient to see how the paths raise in $\phi^{-1}(\gamma)$. These will be paths from some element of $\phi^{-1}(\vec{x})$ to another of $\phi^{-1}(\vec{y})$. These would over-count the paths between $\vec{x}$ and $\vec{y}$ in the original $\mathbb{T}^{d}$; we want to fix a representative in the cover of our origin point. Then, the number of paths in $\mathbb{T}^{d}$ are in bijection with all paths between our fixed representative and all elements of $\phi^{-1}(\vec{y})$. These paths are shown in Figure~\ref{fig:torus}.

\begin{figure}[h]
    \centering
    \begin{tikzpicture}
    \foreach \i in {\xMin,...,\xMax}{
        \draw [very thin,gray] (.2*\i,.2*\yMin) -- (.2*\i,.2*\yMax);
    }
    \foreach \i in {\yMin,...,\yMax} {
        \draw [very thin,gray] (.2*\xMin,.2*\i) -- (.2*\xMax,.2*\i);
    }
    \pgfmathsetmacro{\helperX}{3.0}
    \pgfmathsetmacro{\helperY}{3.0}
    \pgfmathsetmacro{\l}{1}
    \pgfmathsetmacro{\m}{10}
    \foreach \i in {\xMiN,...,\xMaX} {
        \draw [very thick,gray] (2*\i,2*\yMiN) -- (2*\i,2*\yMaX);
        \foreach \j in {\yMiN,...,\yMaX} {
            \ifnum \i<\xMaX
                \ifnum \j<\yMaX
                    \fill[red] (.4+2*\i,.2+2*\j) circle[radius=2pt];
                    %\foreach \k in {0,...,10} {
                    %    \pgfmathsetmacro{\tester}{rand}
                    %    \ifnum\tester<0
                    %        \pgfmathsetmacro{helperXX}{rand}
                    %        \draw [thick,black] (\helperX,\helperY) -- (\helperX+.2,\helperY);
                    %        \pgfmathsetmacro{\helperX}{\helperX+.2}
                    %    \else
                    %        \draw [thick,black] (\helperX,\helperY) -- (\helperX,\helperY+.2);
                    %        \pgfmathsetmacro{\helperY}{\helperY+.2}
                    %    \fi
                    %}
                \fi
            \fi
        }
    }
    \foreach \i in {\yMiN,...,\yMaX} {
        \draw [very thick,gray] (2*\xMiN,2*\i) -- (2*\xMaX,2*\i);
    }
    
    %First Path
    \draw [thick,black] (3.0,3.0) -- (3.0,2.8);
    \draw [thick,black] (3.0,2.8) -- (2.8,2.8);
    \draw [thick,black] (2.8,2.8) -- (2.5999999999999996,2.8);
    \draw [thick,black] (2.5999999999999996,2.8) -- (2.5999999999999996,2.5999999999999996);
    \draw [thick,black] (2.5999999999999996,2.5999999999999996) -- (2.3999999999999995,2.5999999999999996);
    \draw [thick,black] (2.3999999999999995,2.5999999999999996) -- (2.3999999999999995,2.3999999999999995);
    \draw [thick,black] (2.3999999999999995,2.3999999999999995) -- (2.4,2.2);
    
    %Second Path
    \draw [thick,black] (3.0,3.0) -- (3.2,3.0);
\draw [thick,black] (3.2,3.0) -- (3.2,2.8);
\draw [thick,black] (3.2,2.8) -- (3.4000000000000004,2.8);
\draw [thick,black] (3.4000000000000004,2.8) -- (3.6000000000000005,2.8);
\draw [thick,black] (3.6000000000000005,2.8) -- (3.6000000000000005,2.5999999999999996);
\draw [thick,black] (3.6000000000000005,2.5999999999999996) -- (3.8000000000000007,2.5999999999999996);
\draw [thick,black] (3.8000000000000007,2.5999999999999996) -- (4.000000000000001,2.5999999999999996);
\draw [thick,black] (4.000000000000001,2.5999999999999996) -- (4.200000000000001,2.5999999999999996);
\draw [thick,black] (4.200000000000001,2.5999999999999996) -- (4.400000000000001,2.5999999999999996);
\draw [thick,black] (4.400000000000001,2.5999999999999996) -- (4.400000000000001,2.3999999999999995);
\draw [thick,black] (4.400000000000001,2.3999999999999995) -- (4.400000000000001,2.1999999999999993);
    
    %Third Path
    \draw [thick,black] (3.0,3.0) -- (3.2,3.0);
\draw [thick,black] (3.2,3.0) -- (3.4000000000000004,3.0);
\draw [thick,black] (3.4000000000000004,3.0) -- (3.4000000000000004,3.2);
\draw [thick,black] (3.4000000000000004,3.2) -- (3.6000000000000005,3.2);
\draw [thick,black] (3.6000000000000005,3.2) -- (3.6000000000000005,3.4000000000000004);
\draw [thick,black] (3.6000000000000005,3.4000000000000004) -- (3.8000000000000007,3.4000000000000004);
\draw [thick,black] (3.8000000000000007,3.4000000000000004) -- (4.000000000000001,3.4000000000000004);
\draw [thick,black] (4.000000000000001,3.4000000000000004) -- (4.000000000000001,3.6000000000000005);
\draw [thick,black] (4.000000000000001,3.6000000000000005) -- (4.000000000000001,3.8000000000000007);
\draw [thick,black] (4.000000000000001,3.8000000000000007) -- (4.000000000000001,4.000000000000001);
\draw [thick,black] (4.000000000000001,4.000000000000001) -- (4.000000000000001,4.200000000000001);
\draw [thick,black] (4.000000000000001,4.200000000000001) -- (4.400000000000001,4.200000000000001);

    %Fourth Path
    \draw [thick,black] (3.0,3.0) -- (3.0,3.2);
\draw [thick,black] (3.0,3.2) -- (3.0,3.4000000000000004);
\draw [thick,black] (3.0,3.4000000000000004) -- (3.0,3.6000000000000005);
\draw [thick,black] (3.0,3.6000000000000005) -- (3.0,3.8000000000000007);
\draw [thick,black] (3.0,3.8000000000000007) -- (3.0,4.000000000000001);
\draw [thick,black] (3.0,4.000000000000001) -- (2.8,4.000000000000001);
\draw [thick,black] (2.8,4.000000000000001) -- (2.8,4.200000000000001);
\draw [thick,black] (2.8,4.200000000000001) -- (2.5999999999999996,4.200000000000001);
\draw [thick,black] (2.5999999999999996,4.200000000000001) -- (2.3999999999999995,4.200000000000001);
\draw [thick,black] (2.3999999999999995,4.200000000000001) -- (2.4,4.2);
    
    \fill[blue] (3,3) circle[radius=2pt];
    \end{tikzpicture}
    \caption{Sample paths (projected onto spatial dimensions) from $\vec{x}$ to all points in $\phi^{-1}(\vec{y})$ in the $\mathbb{R}^{2}$ cover of $\mathbb{T}^{2}$ for time less than the $\mathbb{T}^{2}$ widths.}
    \label{fig:torus}
\end{figure}
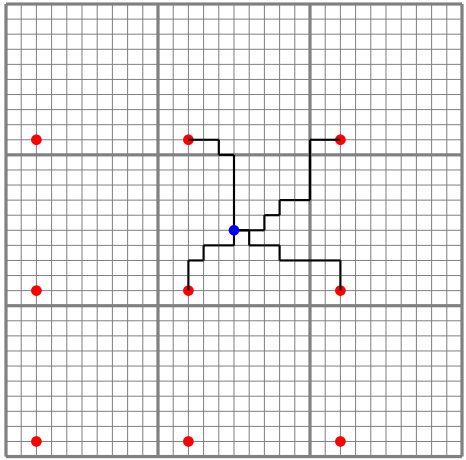

These will be paths in $\mathbb{Z}^{d}\times Z$ between $\vec{x}$ and $\vec{y}'=\vec{y}+\sum_{i}m_{i}2L_{i}e_{i}$. Of course, unless $|\vec{y}'-\vec{x}-proj_{t}(\vec{y}'-\vec{x})\hat{t}|_{l_{1}}\le proj_{t}(\vec{y}'-\vec{x})$ we can never reach $\hat{y}'$. So, our $K_{1}^{\mathbb{T}}$ just becomes the sum of $K_{1}$ for two points displaced as mentioned, which is the expression in the statement of Theorem~\ref{thm:torus}.

\begin{figure}[h]
    \centering
    \begin{tikzpicture}
    \foreach \i in {\xMin,...,\xMax}{
        \draw [very thin,gray] (.2*\i,.2*\yMin) -- (.2*\i,.2*\yMax);
    }
    \foreach \i in {\yMin,...,\yMax} {
        \draw [very thin,gray] (.2*\xMin,.2*\i) -- (.2*\xMax,.2*\i);
    }
    \foreach \i in {\xMiN,...,\xMaX}{
        \draw [very thick,gray] (2*\i,2*\yMiN) -- (2*\i,2*\yMaX);
        \foreach \j in {\yMiN,...,\yMaX} {
            \ifnum \i<\xMaX
                \ifnum \j<\yMaX
                    \pgfmathsetmacro{\helper}{\xMiN + 1.0}
                    \ifnum\i=\helper
                        \fill[red] (.4+2*\i,.2+2*\j) circle[radius=2pt];
                    \else
                        \fill[red] (.4+2*\i,.2+2*\j+1.6) circle[radius=2pt];
                    \fi
                \fi
            \fi
        }
    }
    \foreach \i in {\yMiN,...,\yMaX} {
        \draw [very thick,gray] (2*\xMiN,2*\i) -- (2*\xMaX,2*\i);
    }
    \fill[blue] (3,3) circle[radius=2pt];
    %First Path
    \draw [thick,black] (3.0,3.0) -- (3.0,2.8);
\draw [thick,black] (3.0,2.8) -- (3.0,2.5999999999999996);
\draw [thick,black] (3.0,2.5999999999999996) -- (3.0,2.3999999999999995);
\draw [thick,black] (3.0,2.3999999999999995) -- (2.8,2.3999999999999995);
\draw [thick,black] (2.8,2.3999999999999995) -- (2.5999999999999996,2.3999999999999995);
\draw [thick,black] (2.5999999999999996,2.3999999999999995) -- (2.3999999999999995,2.3999999999999995);
\draw [thick,black] (2.3999999999999996,2.3999999999999995) -- (2.3999999999999995,2.1999999999999995);
    
    %Second Path
    \draw [thick,black] (3.0,3.0) -- (3.0,3.2);
\draw [thick,black] (3.0,3.2) -- (3.0,3.4000000000000004);
\draw [thick,black] (3.0,3.4000000000000004) -- (3.0,3.6000000000000005);
\draw [thick,black] (3.0,3.6000000000000005) -- (3.0,3.8000000000000007);
\draw [thick,black] (3.0,3.8000000000000007) -- (2.8,3.8000000000000007);
\draw [thick,black] (2.8,3.8000000000000007) -- (2.8,4.000000000000001);
\draw [thick,black] (2.8,4.000000000000001) -- (2.5999999999999996,4.000000000000001);
\draw [thick,black] (2.5999999999999996,4.000000000000001) -- (2.3999999999999995,4.000000000000001);
\draw [thick,black] (2.3999999999999995,4.000000000000001) -- (2.3999999999999995,4.200000000000001);
\draw [thick,black] (2.3999999999999995,4.200000000000001) -- (2.4,4.2);
    
    %Third Path
    \draw [thick,black] (3.0,3.0) -- (3.2,3.0);
\draw [thick,black] (3.2,3.0) -- (3.2,3.2);
\draw [thick,black] (3.2,3.2) -- (3.4000000000000004,3.2);
\draw [thick,black] (3.4000000000000004,3.2) -- (3.6000000000000005,3.2);
\draw [thick,black] (3.6000000000000005,3.2) -- (3.8000000000000007,3.2);
\draw [thick,black] (3.8000000000000007,3.2) -- (4.000000000000001,3.2);
\draw [thick,black] (4.000000000000001,3.2) -- (4.000000000000001,3.4000000000000004);
\draw [thick,black] (4.000000000000001,3.4000000000000004) -- (4.200000000000001,3.4000000000000004);
\draw [thick,black] (4.200000000000001,3.4000000000000004) -- (4.400000000000001,3.4000000000000004);
\draw [thick,black] (4.400000000000001,3.4000000000000004) -- (4.400000000000001,3.8000000000000004);
    
    %Fourth Path
    \draw [thick,black] (3.0,3.0) -- (3.2,3.0);
\draw [thick,black] (3.2,3.0) -- (3.2,2.8);
\draw [thick,black] (3.2,2.8) -- (3.2,2.5999999999999996);
\draw [thick,black] (3.2,2.5999999999999996) -- (3.4000000000000004,2.5999999999999996);
\draw [thick,black] (3.4000000000000004,2.5999999999999996) -- (3.4000000000000004,2.3999999999999995);
\draw [thick,black] (3.4000000000000004,2.3999999999999995) -- (3.4000000000000004,2.1999999999999993);
\draw [thick,black] (3.4000000000000004,2.1999999999999993) -- (3.6000000000000005,2.1999999999999993);
\draw [thick,black] (3.6000000000000005,2.1999999999999993) -- (3.8000000000000007,2.1999999999999993);
\draw [thick,black] (3.8000000000000007,2.1999999999999993) -- (3.8000000000000007,1.9999999999999993);
\draw [thick,black] (3.8000000000000007,1.9999999999999993) -- (3.8000000000000007,1.7999999999999994);
\draw [thick,black] (3.8000000000000007,1.7999999999999994) -- (4.400000000000001,1.7999999999999994);

    \end{tikzpicture}
    \caption{Paths from $\vec{x}$ to all points in $\phi^{-1}(\vec{y})$ in the $\mathbb{R}^{2}$ cover of Klein bottle for time less than the $L_{i}$ widths}
    \label{fig:klien}
\end{figure}
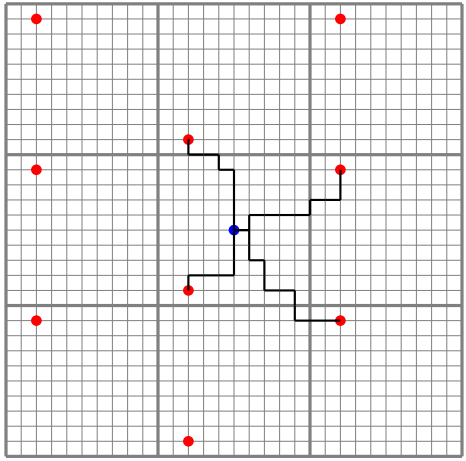

The covering map $\phi':\mathbb{R}^{2}\rightarrow\textrm{klien}$ takes $\vec{x}$ in $\mathbb{R}^{2}$ to its equivalence class in X given in the set $\mathcal{B}$ in the statement of Theorem~\ref{thm:klien}. Then, the exact same proof as above goes to show the desired result. We include Figure~\ref{fig:klien} as the analogous cover to Figure~\ref{fig:torus} but for the Klien bottle. In this manner, finding propagators for spaces that have a covering space $\mathbb{Z}^{d}$ is performed with relative ease (should you describe the equivalence class of points under raising by that cover in $\mathbb{Z}^{d}$ sufficiently well).

\end{proof}

Now that we have shown the ease of computation of $K_{1}$ in the context of free space and orbifolds, we will compute it on tropical surfaces. The surface of chief interest in physics at the moment is de-Sitter space, so we obtain the following proof of Theorem~\ref{thm:deSitter}.

\begin{proof}

Let $\gamma\in \Gamma_{1}$. We note that these are just the same paths as in $\mathbb{Z}^{d}\times\mathbb{Z}$ but constrained to the zero set of $\mathfrak{p}$. This zero set might change $\mathcal{A}_{1}$ (as warned of in the proof of Theorem~\ref{thm:axes}). We note that $n*(0,1)$ is no longer in $\mathcal{A}$. If you advance in time, then to remain on $\mathfrak{p}$, one of your spatial coordinates must change (or $|proj_{t}(\vec{x})|-\sum_{i}|proj_{x_{i}}(\vec{x})|\neq a$ as required of in the hypothesis of Theorem~\ref{thm:deSitter}). Say that $\vec{x}$ is in the positive orthant of $\mathbb{Z}^{d}\times\mathbb{Z}$. Then, $|proj_{t}(\vec{x})|-\sum_{i}|proj_{x_{i}}(\vec{x})|=a\implies |proj_{t}(\vec{x}+\hat{t})|-\sum_{i}|proj_{x_{i}}(\vec{x}+e_{i})|=a$ and we have shown that $(e_{i},1)$ still must be in $\mathcal{A}_{1}$. Our proof in Theorem~\ref{thm:axes} demonstrates that $\mathcal{A}$ can now only include the null vectors, and so  we have shown that in this context $\mathcal{A}_{1}$ is $\{(e_{i},1)|e_{i}\textrm{ is a unit vector in }\mathbb{Z}^{d}\}$. So, by Theorem~\ref{thm:generate}, we may consider $\gamma$ as having a difference sequence in $\mathcal{A}_{1}$.

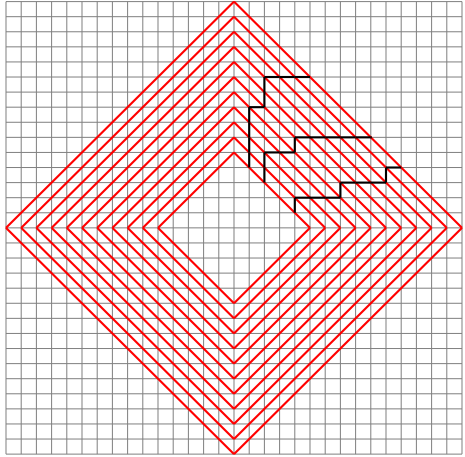
\begin{figure}[h]
    \centering
    \begin{tikzpicture}
    \foreach \i in {\xMin,...,\xMax}{
        \draw [very thin,gray] (.2*\i,.2*\yMin) -- (.2*\i,.2*\yMax);
    }
    \foreach \i in {\yMin,...,\yMax} {
        \draw [very thin,gray] (.2*\xMin,.2*\i) -- (.2*\xMax,.2*\i);
    }
    \foreach \t in {0,...,10}{
        \draw [thick,red] (2.0-.2*\t,3.0) -- (3.0,2.0-.2*\t);
        \draw [thick,red] (3.0,2.0-.2*\t) -- (4.0+.2*\t,3.0);
        \draw [thick,red] (4.0+.2*\t,3.0) -- (3.0,4.0+.2*\t);
        \draw [thick,red] (3.0,4.0+.2*\t) -- (2.0-.2*\t,3.0);
    }
    
    \draw [thick,black] (3.4,3.6) -- (3.4,3.8000000000000003);
\draw [thick,black] (3.4,3.8000000000000003) -- (3.4,4.0);
\draw [thick,black] (3.4,4.0) -- (3.6,4.0);
\draw [thick,black] (3.6,4.0) -- (3.8000000000000003,4.0);
\draw [thick,black] (3.8000000000000003,4.0) -- (3.8000000000000003,4.2);
\draw [thick,black] (3.8000000000000003,4.2) -- (4.0,4.2);
\draw [thick,black] (4.0,4.2) -- (4.2,4.2);
\draw [thick,black] (4.2,4.2) -- (4.4,4.2);
\draw [thick,black] (4.4,4.2) -- (4.6000000000000005,4.2);
\draw [thick,black] (4.6000000000000005,4.2) -- (4.800000000000001,4.2);
\draw [thick,black] (4.800000000000001,4.2) -- (4.8,4.2);
    
\draw [thick,black] (3.2,3.8) -- (3.2,4.0);
\draw [thick,black] (3.2,4.0) -- (3.2,4.2);
\draw [thick,black] (3.2,4.2) -- (3.2,4.4);
\draw [thick,black] (3.2,4.4) -- (3.2,4.6000000000000005);
\draw [thick,black] (3.2,4.6000000000000005) -- (3.4000000000000004,4.6000000000000005);
\draw [thick,black] (3.4000000000000004,4.6000000000000005) -- (3.4000000000000004,4.800000000000001);
\draw [thick,black] (3.4000000000000004,4.800000000000001) -- (3.4000000000000004,5.000000000000001);
\draw [thick,black] (3.4000000000000004,5.000000000000001) -- (3.6000000000000005,5.000000000000001);
\draw [thick,black] (3.6000000000000005,5.000000000000001) -- (3.8000000000000007,5.000000000000001);
\draw [thick,black] (3.8000000000000007,5.000000000000001) -- (4.000000000000001,5.000000000000001);

\draw [thick,black] (3.8,3.2) -- (3.8,3.4000000000000004);
\draw [thick,black] (3.8,3.4000000000000004) -- (4.0,3.4000000000000004);
\draw [thick,black] (4.0,3.4000000000000004) -- (4.2,3.4000000000000004);
\draw [thick,black] (4.2,3.4000000000000004) -- (4.4,3.4000000000000004);
\draw [thick,black] (4.4,3.4000000000000004) -- (4.4,3.6000000000000005);
\draw [thick,black] (4.4,3.6000000000000005) -- (4.6000000000000005,3.6000000000000005);
\draw [thick,black] (4.6000000000000005,3.6000000000000005) -- (4.800000000000001,3.6000000000000005);
\draw [thick,black] (4.800000000000001,3.6000000000000005) -- (5.000000000000001,3.6000000000000005);
\draw [thick,black] (5.000000000000001,3.6000000000000005) -- (5.000000000000001,3.8000000000000007);
\draw [thick,black] (5.000000000000001,3.8000000000000007) -- (5.200000000000001,3.8000000000000007);

    \end{tikzpicture}
    \caption{Red diamonds represent the curves $|x|+|y|-|t|=5$ for $t\in\{0,...,10\}$, and paths upon 1d de-Sitter space from $t=0\rightarrow t=10$}
    
    %Paths from $\vec{x}$ to all points in $\phi^{-1}(\vec{y})$ in the $\mathbb{R}^{2}$ cover of Klein bottle for time less than the $L_{i}$ widths}
    \label{fig:sitter}
\end{figure}

So, $\gamma$ may only move monotonically in any direction in $\mathbb{Z}^{d}$. This is demonstrated in Figure~\ref{fig:sitter} for $d=1$. Since it's composed only of the null paths, $\rho_{n}(\gamma)$ is necessarily zero. The number of monotonically moving lattice paths is the multinomial in the statement of Theorem~\ref{thm:deSitter}. This is well known, refer to \cite{continuouslatticepath};  you can separate each step in your path $\vec{x}\rightarrow\vec{y}$ into sets based on which direction they moved. This will partition the time steps $proj_{t}(\vec{y}-\vec{x})$, so we obtain the combinatorial result. The continuum result is obtained from Theorem~\ref{thm:contconv}.

%Consider $x_{i+1}-x_{i}$, an element of our difference sequence. 

%\note{show that only half of the walks}

%Therefore 

%The primary difference is that the magnitude of $\mathcal{A}_{p}$ may be larger. We will assume first that $\vec{y},\vec{x}$ are together in the positive orthant of $\mathbb{Z}^{d}\times \mathbb{Z}$.

\end{proof}

With the rigorous proofs included, I now conclude with Section~\ref{sec:disc}, which includes conjectures and computational work with interaction terms in our propagators.

\section{Computational Work and Conjectures}
\label{sec:disc}

\subsection{Mathematical Conjectures}
In this section, we include computational work and conjectures that will be addressed in future work. First, we would like to present an immediate conjecture from our work above. In proving Theorem~\ref{thm:k2}, we obtained Equation~\ref{equation:PythagoreanTriples} for $x,t\in \mathbb{R}_{+}\textrm{ and }t>x$.

\begin{multline}
    \label{equation:PythagoreanTriples}
    lim_{n\rightarrow\infty}\int...\int \begin{Bmatrix}t-\sum_{k\in\mathcal{A}_{timeless}'}(proj_{t}(\vec{a}_{k})-1)I_{k}\\\{I_{k}\}_{k\in\mathcal{A}_{timeless}'},I-\sum_{k\in\mathcal{A}_{timeless}'}d_{n}(0,\vec{a}_{k})I_{k},f_{\pm}\end{Bmatrix}\Pi_{k\in\mathcal{A}_{timeless}'dI_{k}}
    \\
    =\sqrt{I^{2}+t^{2}-x^{2}}
    \\
    f_{\pm}=\frac{t^{2}-x^{2}-\sum_{k\in\mathcal{A}_{timeless}'}(proj_{t}(\vec{a}_{k})\pm proj_{x}(\vec{a}_{k})-d_{n}(0,\vec{a}_{k}))I_{k}-I}{2}
\end{multline}

where $\mathcal{A}_{timeless}'$ is $\mathcal{A}_{n}\setminus\{(-1,1),(0,1),(1,0)\}$ with the added constraint that we include negative $\rho_{DB}$ corresponding to a 'anti-particle.' This is explained in Section~\ref{section:intro}. The result of Theorem~\ref{thm:freeprop} seems to suggest (if pythagorean tuples are equidistributed for all dimensions) that a similar result would be true for all dimensions. Specifically, let $\mathcal{A}^{d}_{n}$ be the set of all primitive Pythagorean tuples \cite{rational} of dimension d+1 with radii less than n. Let us index $\vec{a}\in\mathcal{A}^{d}$ such that $proj_{x_{0}}(\vec{a})$ is the hypotenuse, and fix a constant index for the other coordinates. We may also denote as $proj_{t}(\vec{a})$ the hypotenuse of the tuple, defining $x_{0}$ as a time-like direction. We let $proj_{x_{d+1}}(\vec{a})=d_{l_{2}^{*}}^{d}(0,\vec{a})$. It is necessarily the case that one of the non-hypotenuse tuple elements satisfies $d_{l_{2}^{*}}^{d}(0,\vec{a})=\sqrt{proj_{t}(\vec{a})^{2}-\sum_{i=1}^{d}proj_{x_{i}}(\vec{a})^{2}}$, and so we fix it as the last coordinate. Then, necessarily it always contains $(1,\vec{0})$ and all vectors from Theorem~\ref{thm:axes}. We denote as $\mathcal{A}_{timeless}'$ the set $\mathcal{A}_{n}^{d}$ set minus the vectors from Theorem~\ref{thm:axes} with the negative phases as used in the definition of $K_{l_{2}^{*}}^{Feyn}$. Altogether, Theorem~\ref{thm:freeprop} and Theorem~\ref{thm:k2}, along with the notion that our approach rigorously yields the relativistic scalar field propagator \cite{Hong_Hao_2010}, suggests Equation~\ref{equation:PythagoreanTupples}.

\begin{equation}
\label{equation:PythagoreanTupples}
    lim_{n\rightarrow\infty}K_{n}^{Feyn,cont}(\vec{x},\vec{y})=C\mathcal{F}|_{I}^{m}((1-\frac{I^{2}+\sum_{i=1}^{d}x_{i}^{2}}{t^{2}})^{\frac{1-2d}{2}})
\end{equation}

Here, $K_{n}^{Feyn}$ will be defined as in Section~\ref{section:intro} (the limit of $\mathcal{T}^{cont}$ of a lattice path sum over higher dimensional paths), and we should derive a similarly insightful expression for it as a limit of continuum multinomial coefficients. We will not write this expression for the sake of expediency. The relationship between the Pythagorean tuples and rational solutions to the equation $d_{l_{2}^{*}}(0,\vec{a})=1$ \cite{rational} as well as the extension to select Riemann Surfaces suggests that lattice path integrals can be of use towards algebraic geometry. We could try to discretize the continuum lattice path integral as was done in \cite{continuousbin} to study rational points on general surfaces.

%Finally the author believes this work has an interesting expression as a statement about powers of the adjacency matrix of the Pythagorean triples. Consider a graph $G=(\mathbb{Z}^{2},\mathcal{E})$ where two vertices $v_{1},v_{2}\in \mathbb{Z}^{2}$ correspond to an edge $v_{1}v_{2}\in\mathcal{E}$ if $\vec{v}_{1}-\vec{v}_{2}=(x,I)$ is a Pythagorean triple (i.e. $\sqrt{x^{2}+I^{2}}\in\mathbb{N}$). Let $\mathcal{M}_{\vec{i}\vec{j}}=(proj_{I}(\vec{v}_{i}-\vec{v}_{j}))\mathbbm{1}_{\vec{v}_{i},\vec{v}_{j}}$ be our adjacency matrix weighted by the phase portion of our Pythagorean triple. Powers of $\mathcal{M}$ then give the weighted sum of all 

\subsection{Interactions}

Now, I would like to move to a discussion of interactions in this picture. The path-indexed sum over phases, along with $\mathcal{T}_{cont}^{m}$, suggests a computational way to compute the above interactions. We can generate all paths (choosing those will low numbers of distinct linear segments) and perform the phase-indexed sum over them. These computations, however, also allow us to consider introducing an interaction term. We first will consider our particle moving w.r.t a static charged particle (static in our observer's relativistic frame). For $1-1$ space, we use known results for Coulomb gases \cite{Lewin_2022} and have $V_{coul}$ in Equation~\ref{eq:coulCharge}.

\begin{equation}
\label{eq:coulCharge}
    V_{coul}(x)=-|x-x_{q}|
\end{equation}

This potential has the property that $\nabla V_{coul}=\delta(x-x_{q})$, which implies that it is the analog of an electric charge potential in 1 spatial dimension. We can add it to our action in the natural way; action is classically kinetic minus potential energy \cite{pesky}. Since our isoperimetry in Section~\ref{section:def} covers kinetic energy, we simply add $-im\sum_{\vec{x}_{i}\in\gamma}V_{coul}(\vec{x}_{i})$ to our phase. Here we note $\vec{x}_{i}$ does not refer to the difference sequence of $\gamma$, but rather its moving endpoint. When we desire our continuous quantity, we must normalize the units of our potential. If $\rho$ is normalized by division by $n$, it makes sense to normalize $V_{coul}$ by division by $n^{2}$. This is because it is a sum over $t\sim n$ terms of magnitude $\vec{x}_{i}\sim n$ (not a difference sequence), meaning its scaling ought to be as $n^{2}$; Figure~\ref{fig:freePart} shows the free particle $l_{1}$ propagator at 12 seconds, while Figure~\ref{fig:electricField} demonstrates the propagator for a 1d electric charge located at $x=1$ light-second, $m=1$, and at $t=2$ seconds.

\begin{figure}[h]
    \centering
    \includegraphics[width = 6 cm]{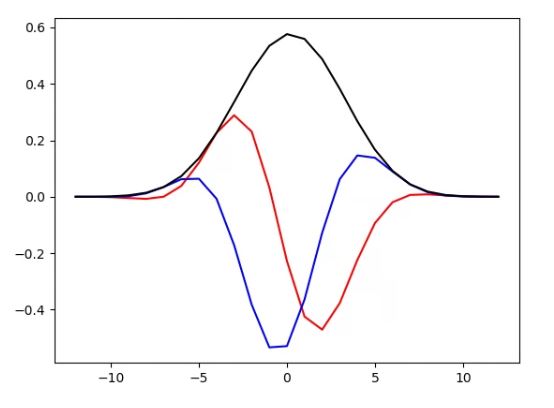}
    \caption{Free Taxicab Propogator at $t=12.0$ seconds}
    \label{fig:freePart}
\end{figure}

\begin{figure}[h]
    \centering
    \includegraphics[width = 6 cm]{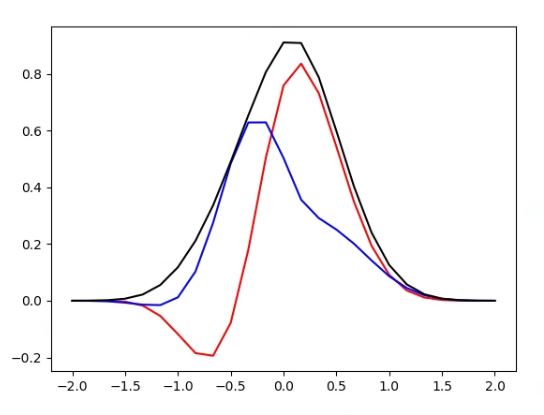}
    \caption{Attractive Coulomb Potential at $x=1.0$ light seconds after $t=2.0$ seconds}
    \label{fig:electricField}
\end{figure}

In Figure~\ref{fig:freePart} and Figure~\ref{fig:electricField}, the red is real, the blue imaginary, and the black is the magnitude. Clearly, as time moves on, the probability amplitude of our scalar particle moves towards the attractive coulomb well. We can be more sure we are performing this simulation correctly by taking an inverse Fourier transform of the complex-valued propagator and probing it for spikes in $m$. The author believes that these spikes would correspond to bound or resonant states of the above system.

\subsection{Discussion}

The above interactions are obtained by considering different, non-geometric terms in our action.  It is difficult to obtain physically interesting interaction terms from purely geometric quantities; this is one of the promising features of String and M-Theory \cite{vafa2005string}. In these theories, we consider sheets and membranes rather than paths, and the quantities of interest are discrete sheet-indexed sums of phases. Like in our above work, these phases are determined by the area of the sheets, and interactions are directly built into high-genus sheets. There are a number of technical hurdles with using our above scheme to find a Membrane Theory. The first is that our input and output states are no longer points with finite degrees of freedom; string and membrane initial and final states have infinitely many degrees of freedom. The author has already obtained some results for discrete $l_{1}$ string propagators when the degrees of freedom are limited by allowing world-sheets (the sheets traced out by strings in String Theory \cite{vafa2005string}) to only be rectangular. Another means of solving this problem would be to treat the Membrane propagators as a function between sequences, where a sequence would encode uniquely the moments of a string or membrane. In either case, the work of Cano and Diaz \cite{continuousbin} must be generalized to obtain Membrane Theory's version of $\mathcal{T}_{m}^{cont}$. It is not so easily seen that the space of directed sheets has some sensible expression for its volume; this would be necessary to move from discrete string propagators (of which the author has already performed some computation) to the continuous case.

%and fixing an index $a_{i}$ for the other sides of our Pythagorean tupple we let $d_{n}(0,\vec{a})$ denote $a_{1}$

%Section~\ref{section:def} includes definitional proofs which do not fit thematically in any place in our above paper. We place these in a separate section because while they are critical for rigor, they are either basic or require in-depth calculation.

\section{Definitional Proofs}
\label{section:def}

This section includes longer and more difficult proofs of concepts required throughout this work. Some concepts are results in themselves, and therefore ought not to be relegated to an appendix.

\begin{theorem}
\label{thm:multiNomBound}

Let $\{x_{i}\}_{i=1}^{l}\subset\mathbb{N}$ such that $min(x_{i}-\frac{\sum_{i=1}^{l}x_{i}}{l})\ge n$

$$\begin{pmatrix}
\sum_{i=1}^{l}x_{i}\\x_{1},...,x_{l}
\end{pmatrix}= \frac{l^{\sum_{i=1}^{l}x_{i}+\frac{l}{2}}}{\sqrt{2\pi \sum_{i=1}^{l}x_{i}}^{l-1}}e^{-\frac{l}{2\sum_{i=1}^{l}x_{i}}(\sum_{i=1}^{l}(x_{i}-\frac{\sum_{i=1}^{l}x_{i}}{l})^{2})+o(1)}$$

where $o(1)$ denotes some function that goes to zero as $n\rightarrow \infty$. An alternative approximation is the following

$$\sim e^{-\sum_{i=1}^{l}x_{i}ln(\frac{x_{i}}{\sum_{i=1}^{1}x_{i}})+o(\sum_{i=1}^{l}x_{i})}$$

relating the continuum multinomial to Shannon entropy.

\end{theorem}

This theorem is given in \cite{344669} and the shannon entropy relationship in \cite{vigneaux2020homological}

\begin{theorem}
\label{thm:contconv}
The continuous multinomial

$$\begin{Bmatrix}\sum x_{i}\\x_{1},x_{2},...,x_{l}\end{Bmatrix}<\infty$$

is a finite and analytic function from $\mathbb{R}^{l}\rightarrow\mathbb{R}$.

\end{theorem}
\begin{proof}

In \cite{continuouslatticepath}, they demonstrate that Eq~\ref{eq:contmult} becomes the following

\begin{equation}
    \label{eq:contmults1}
    \begin{Bmatrix}\sum x_{i}\\x_{1},x_{2},...,x_{l}\end{Bmatrix}=\sum_{\eta_{1},...,\eta_{l}=0}^{\infty}(f_{\nu_{1},...,\nu_{l}}\Pi_{k=1}^{l}(\frac{\sqrt{\nu_{k}}x_{k}^{\nu_{k}}}{\nu_{k}!}))
\end{equation}

where $f_{v_{1},...,v_{n}}$ counts Smirnov words with frequency vector $v_{i}$. This vector denotes how many of each letter occurs in a Smirnov word. This coefficient has a nice generation function, and in \cite{continuouslatticepath} it was shown that Eq~\ref{eq:contmults1} without the root terms has nice analytic properties. Namely, it satisfies Equation~\ref{eq:contmults2}. 

\begin{equation}
    \label{eq:contmults2}
    (\Pi_{i=1}^{l}(1+\partial_{x_{i}})-\sum_{i=1}^{l}\Pi_{j\neq i}(1+\partial_{x_{j}}))\begin{Bmatrix}\sum x_{i}\\x_{1},x_{2},...,x_{l}\end{Bmatrix}=0
\end{equation}

We will use this to obtain an absolutely convergent multidimensional taylor series at 0 for the continuous multinomial (with an infinite radius of convergence). This will demonstrate the claim. Suppose we let $a_{i_{1},...,i_{l}}\in \mathbb{R}$ denote the taylor series coefficients for the continuous multinomial; i.e. we have Equation~\ref{eq:contmults3}.

\begin{equation}
    \label{eq:contmults3}
    \begin{Bmatrix}\sum x_{i}\\x_{1},x_{2},...,x_{l}\end{Bmatrix}=\sum a_{i_{1},...,i_{l}}\Pi_{k=1}^{l}x_{k}^{i_{k}}
\end{equation}

We note that from the work of \cite{continuouslatticepath} (expressing Equation~\ref{eq:contmults1} sans square root as a Borel transform) that the following coefficients are immediate:

$$a_{0,...,0}=1$$

The non-mixed terms for $x_{1}$ in Equation~\ref{eq:contmults2} are first order; they give us a recursion relation for $a_{i,0,...,0}$:

$$((1-l)+(2l-l^{2})\partial_{x_{1}})\sum_{i_{1}=1}a_{i_{1},...,0}x_{1}^{i_{1}}=0$$

This gives us the recursion relation $a_{i,...,0}=\frac{(l-1)}{(2l-l^{2})i}a_{i-1,...,0}$ with the solution:

$$a_{i,...,0}=\frac{(\frac{l-1}{2l-l^{2}})^{i}}{i!}$$

for $i>0$. This has an infinite radius of convergence in its one direction. To proceed further, we acknowledge the symmetries of the function, and our differential equation gives us $a_{i_{1},...,i_{l}}=a_{\sigma(i_{1},...,i_{l})}$ where $\sigma\subset S_{l}$ is an arbitrary permutation of our indices. So, we need only solve $l$ consecutive partial recursion relations to obtain our full solution. Namely, we need a recursion relation for arbitrary terms of the form $a_{i_{1},...,i_{j},0,..,0}$, where terms like $a_{i_{1},...,i_{j-1},0,..,0}$ provide boundary conditions. From Equation~\ref{eq:contmults2}, we can list the lowest order differential operators with terms with $\{i_{1},...,i_{j}\}$ in them

$$(\begin{pmatrix}l\\j\end{pmatrix}-(l-j)\begin{pmatrix}l-1\\j\end{pmatrix})\Pi_{k=1}^{i}\partial_{x_{k}}$$

which, when applied to our taylor series, gives us terms as in Equation~\ref{eq:middlePart}.

\begin{equation}
    \label{eq:middlePart}
    (\begin{pmatrix}l\\j\end{pmatrix}-(l-j)\begin{pmatrix}l-1\\j\end{pmatrix})(\Pi_{k=1}^{j}i_{k})a_{i_{1},i_{2},...,i_{j},0,...,0}
\end{equation}

Now, to find the partial recurrence relation for $a_{i_{1},i_{2},...,i_{j},0,...,0}$, we need to combine all terms in Equation~\ref{eq:contmults1}, which have only differential operators from $i_{1}\rightarrow i_{j}$. If we are missing an index, then that corresponds to an $a$ term with the decremented index itself. Each of these contributions to the differential operators will look like the expression from Equation~\ref{eq:middlePart}. This gives us the partial recursion relation Equation~\ref{eq:finalRecurse} for $a_{i_{1},...,i_{j},0,..,0}$, where $Q$ runs over the number of indices missing from some differential term.

\begin{multline}
    \label{eq:finalRecurse}    
    (\begin{pmatrix}l\\j\end{pmatrix}-(l-j)\begin{pmatrix}l-1\\j\end{pmatrix})(\Pi_{k=1}^{j}i_{k})a_{i_{1},i_{2},...,i_{j},0,...,0}
    \\
    +\sum_{Q=1}^{l}\sum_{\{q_{i}\}_{i=1}^{Q},q_{i}\in \{1,...,j\},q_{i}=q_{j}\implies i=j}(C_{l,j,Q})(\Pi_{k=1,k\not\in\{q_{1},..., q_{Q}\}}i_{k})a_{i_{1},...,i_{q_{1}}-1,...,i_{q_{Q}}-1,...,i_{j},...,0}=0   
    \\
    \textrm{ where }C_{l,j,Q}=\begin{pmatrix}l\\j-Q\end{pmatrix}-(l-j+Q)\begin{pmatrix}l-1\\j-Q\end{pmatrix}
\end{multline}

This gives us an algorithm to calculate $a_{i_{1},...,i_{l}}$ for any index, by running through a higher order recursion relation until we hit zero in some index of $a$, and then moving to a lower recursion relation until we reach $a_{0,...,0}=1$. Our goal now is to use Equation~\ref{eq:finalRecurse} to bound the taylor terms into some absolutely convergent series as we have for $a_{i,0,...,0}$ when we evaluate the continuous multinomial along only a single coordinate. Rewriting Equation~\ref{eq:finalRecurse}, we obtain Equation~\ref{eq:simpleRecurse}.

\begin{multline}
    \label{eq:simpleRecurse}    
    a_{i_{1},i_{2},...,i_{j},0,...,0}
    \\
    =-\sum_{Q=1}^{l}\sum_{\{q_{i}\}_{i=1}^{Q},q_{i}\in \{1,...,j\},q_{i}=q_{j}\implies i=j}(\frac{C_{l,j,Q}}{C_{l,j,0}})(\Pi_{k=1,k\in\{q_{1},..., q_{Q}\}}i_{k}^{-1})a_{i_{1},...,i_{q_{1}}-1,...,i_{q_{Q}}-1,...,i_{j},...,0}   
    \\
    \textrm{ where }C_{l,j,Q}=\begin{pmatrix}l\\j-Q\end{pmatrix}-(l-j+Q)\begin{pmatrix}l-1\\j-Q\end{pmatrix}
\end{multline}

Now, we can use the above recursion relation to obtain bounds on the taylor series terms. We have a first bound in Equation~\ref{eq:simpleRecurse2}.
\tiny
\begin{multline}
    \label{eq:simpleRecurse2}    
    |a_{i_{1},i_{2},...,i_{j},0,...,0}|
    \\
    \le \sum_{Q=1}^{l}\sum_{\{q_{i}\}_{i=1}^{Q},q_{i}\in \{1,...,j\},q_{i}=q_{j}\implies i=j}\left|(\frac{C_{l,j,Q}}{C_{l,j,0}})\right|(\Pi_{k=1,k\in\{q_{1},...,q_{Q}\}}i_{k}^{-1})\left |a_{i_{1},...,i_{q_{1}}-1,...,i_{q_{Q}}-1,...,i_{j},...,0}\right|   
    \\
    \textrm{ where }C_{l,j,Q}=\begin{pmatrix}l\\j-Q\end{pmatrix}-(l-j+Q)\begin{pmatrix}l-1\\j-Q\end{pmatrix}
\end{multline}
\normalsize
Exhaustively applying our recursion relation Equation~\ref{eq:simpleRecurse}, we obtain Equation~\ref{eq:simpleRecurse3}.
\tiny
\begin{multline}
    \label{eq:simpleRecurse3}    
    |a_{i_{1},i_{2},...,i_{j},0,...,0}|\le
    \\
    \frac{1}{\Pi_{k=1}^{j}i_{k}!}\sum_{Q=1}^{l}\sum_{\{q_{i}\}_{i=1}^{Q},q_{i}\in \{1,...,j\},q_{i}=q_{j}\implies i=j}\left|(\frac{C_{l,j,Q}}{C_{l,j,0}})\right|(\sum_{Q=1}^{l}(\sum_{\{q_{i}\}_{i=1}^{Q},q_{i}\in \{1,...,j'\},q_{i}=q_{j}\implies i=j}\left|(\frac{C_{l,j',Q}}{C_{l,j',0}})\right|...)))   
    \\
    \textrm{ where }C_{l,j,Q}=\begin{pmatrix}l\\j-Q\end{pmatrix}-(l-j+Q)\begin{pmatrix}l-1\\j-Q\end{pmatrix}
\end{multline}
\normalsize
Each of these sums ends in $a_{0,...,0}=1$. The fraction of $C_{l,j,Q}$ can be bounded absolutely by a constant in terms of l, and we can consume each intermediate double sum by that constant as well. Let us call that constant $C_{l}$. Then, we have a factor $C_{l}^{\Pi_{k=1}^{j}i_{k}}$ bounding this whole series, giving us the final bound in Equation~\ref{eq:finalBound}.

\begin{equation}
    \label{eq:finalBound}    
    |a_{i_{1},i_{2},...,i_{j},0,...,0}|\le \frac{C_{l}^{\Pi_{k=1}^{j}i_{k}}}{\Pi_{k=1}^{j}i_{k}!}
\end{equation}

So, this nice modified expression is analytic and finite at all values. The original expression for the continuous multinomial coefficient (in terms of volumes) has an extra factor of $\Pi_{k=1}^{l}\sqrt{i_{k}}$ that multiplies each $a_{i_{1},...,i_{l}}$. It is clear from the above work that this would not affect the infinite radius of convergence, and we have demonstrated the well-definedness of our continuous multinomial coefficient. Since our series converges absolutely, the other expression for the multinomial (in Equation~\ref{eq:contmult}) may be a permutation of the series obtained above. It will converge to the same value because of our absolute convergence, and we can alternatively adopt either expression. It shall be important to adopt the more geometric expression in the coming proof.

\end{proof}

\begin{theorem}{The Continuous Multinomial is a Limit of the Discrete One}
\label{thm:disctocont}

Let $\{x_{1}\}_{i=1}^{l}\subset\mathbb{R}_{+}$. Then we have

$$lim_{m\rightarrow\infty}\frac{\mathcal{T}^{m}_{cont}\begin{pmatrix}\sum [m x_{i}]\\ [mx_{1}],..., [mx_{l}]\end{pmatrix}}{\mathcal{T}^{m}_{cont}\begin{pmatrix}\sum_{i}\lfloor nx_{i}\rfloor\\\lfloor \frac{\sum_{i} nx_{i}}{l}\rfloor,...,\lfloor\frac{\sum_{i} nx_{i}}{l}\rfloor\end{pmatrix}}\rightarrow \frac{\begin{Bmatrix}\sum_{i}x_{i}\\x_{1},...,x_{l}\end{Bmatrix}}{\begin{Bmatrix}\sum_{i}x_{i}\\ \frac{\sum_{i} nx_{i}}{l},...,\frac{\sum_{i} nx_{i}}{l}\end{Bmatrix}}$$

This will also prove other forms of convergence (here we normalize by the maxima, we could also normalize by the integral of each expression).

\end{theorem}
\begin{proof}

Consider the polyhedron $P(q,c)$ used in the expression in Equation~\ref{eq:contmult}. As a polyhedron, it lies in the intersection of Lebesgue and Reimmanian measurable sets, so we have Equation~\ref{eq:limit}.

\begin{equation}
    \label{eq:limit}
    \mu(P(q,c))=lim_{m\rightarrow\infty}(|\{\vec{v}\in\mathbb{Z}^{n-d}\textrm{ and }\frac{\vec{v}}{m}\in P(q,c)\}|*(\frac{1}{m})^{n-d})
\end{equation}

The fact that $P(q,c)$ is an $n-l$ dimensional polyhedron, where n refers to the number of steps allowed in a directed path, is a fact made apparent by \cite{continuouslatticepath}. We will just use it here. Let $\{x_{i}\}_{i=1}^{l}\subset \mathbb{N}$. We note that the multinomial function is well known to be equal to the number of monotonic paths with steps among $\{e_{i}\}_{i=1}^{l}$ from zero to $\sum x_{i}$. Therefore, we would have Equation~\ref{eq:limit2}

\begin{equation}
    \label{eq:limit2}
    \begin{pmatrix}
    \sum x_{i}\\x_{1},...,x_{l}
    \end{pmatrix}=\sum_{n=0}^{\sum x_{i}-1}\sum_{c\in D(n,l)}|\{\vec{v}\in\mathbb{Z}^{n-d}\cap P(\sum_{i=1}^{l}x_{i}\vec{e}_{i},c)\}|
\end{equation}

Now, let $\{x_{i}\}_{i=1}^{l}\in \mathbb{R}_{+}$. Plugging in $m x_{i}$ into Equation~\ref{eq:limit2} we get Equation~\ref{eq:limit3}.

\begin{equation}
    \label{eq:limit3}
    \begin{pmatrix}\sum [m x_{i}]\\ [mx_{1}],..., [mx_{l}]\end{pmatrix}
    =\sum_{n=0}^{\sum [m x_{i}]-1}\sum_{c\in D(n,l)}|\{\vec{v}\in\mathbb{Z}^{n-d}\cap P(\sum_{i=1}^{l}[mx_{i}]\vec{e}_{i},c)\}|
\end{equation}

This becomes Equation~\ref{eq:limit4}.

\begin{equation}
    \label{eq:limit4}
    \begin{pmatrix}\sum [m x_{i}]\\ [mx_{1}],..., [mx_{l}]\end{pmatrix}
    =m^{-d}\sum_{n=0}^{\sum [m x_{i}]-1}(\sum_{c\in D(n,l)}\frac{|\{\vec{v}\in\mathbb{Z}^{n-l}\cap P(\sum_{i=1}^{l}[mx_{i}]\vec{e}_{i},c)\}|}{m^{n-d}})m^{n}
\end{equation}

Now, we want to introduce the fraction we have from the start of this theorem's statement. When we do so, Equation~\ref{eq:limit4} becomes Equation~\ref{eq:limit5}.

\begin{equation}
    \label{eq:limit5}
    \frac{\begin{pmatrix}\sum [m x_{i}]\\ [mx_{1}],..., [mx_{l}]\end{pmatrix}}{\begin{pmatrix}\sum_{i}\lfloor nx_{i}\rfloor\\\lfloor \frac{\sum_{i} nx_{i}}{l}\rfloor,...,\lfloor\frac{\sum_{i} nx_{i}}{l}\rfloor\end{pmatrix}}
    =\frac{\sum_{n=0}^{\sum [m x_{i}]-1}(\sum_{c\in D(n,l)}\frac{|\{\vec{v}\in\mathbb{Z}^{n-l}\cap P(\sum_{i=1}^{l}[mx_{i}]\vec{e}_{i},c)\}|}{m^{n-l}})m^{n}}{\sum_{n=0}^{\sum [m x_{i}]-1}(\sum_{c\in D(n,l)}\frac{|\{\vec{v}\in\mathbb{Z}^{n-l}\cap P(\sum_{j=1}^{l}\frac{(\sum_{i=1}^{l}[mx_{i}])}{l}\vec{e}_{j},c)\}|}{m^{n-l}})m^{n}}
\end{equation}

Employing Equation~\ref{eq:limit} into Equation~\ref{eq:limit5}, we get Equation~\ref{eq:limit6}

\begin{equation}
    \label{eq:limit6}
    \frac{\begin{pmatrix}\sum [m x_{i}]\\ [mx_{1}],..., [mx_{l}]\end{pmatrix}}{\begin{pmatrix}\sum_{i}\lfloor nx_{i}\rfloor\\\lfloor \frac{\sum_{i} nx_{i}}{l}\rfloor,...,\lfloor\frac{\sum_{i} nx_{i}}{l}\rfloor\end{pmatrix}}
    \rightarrow \frac{\sum_{n=0}^{\sum [m x_{i}]-1}(\sum_{c\in D(n,l)}\mu P(\sum x_{i}e_{i},c))m^{n}}{\sum_{n=0}^{\sum [m x_{i}]-1}(\sum_{c\in D(n,l)}\mu P(\sum_{j=1}^{l}\frac{\sum_{i=1}^{l}x_{i}}{l}\vec{e}_{j},c))m^{n}}
\end{equation}

We apply the series transform $\mathcal{T}_{cont}$ described in Section~\ref{section:def} to this series; in order to reweight the series back towards less segmented paths. This gives us Equation~\ref{eq:limit7}.

\begin{equation}
    \label{eq:limit7}
    lim_{m\rightarrow\infty}\frac{\mathcal{T}^{m}_{cont}\begin{pmatrix}\sum [m x_{i}]\\ [mx_{1}],..., [mx_{l}]\end{pmatrix}}{\mathcal{T}^{m}_{cont}\begin{pmatrix}\sum_{i}\lfloor nx_{i}\rfloor\\\lfloor \frac{\sum_{i} nx_{i}}{l}\rfloor,...,\lfloor\frac{\sum_{i} nx_{i}}{l}\rfloor\end{pmatrix}}
    \rightarrow \frac{\begin{Bmatrix}\sum_{i}x_{i}\\x_{1},...,x_{l}\end{Bmatrix}}{\begin{Bmatrix}\sum_{i}x_{i}\\ \frac{\sum_{i} x_{i}}{l},...,\frac{\sum_{i} x_{i}}{l}\end{Bmatrix}}
\end{equation}

\end{proof}

\begin{theorem}
\label{thm:conttails}

Let $\{x_{i}\}_{i=1}^{l}\in \mathbb{R}$ where $min(x_{i}-\frac{\sum_{i=1}^{n}x_{i}}{l})>R\in\mathbb{R}_{+}$. Then

%$$\begin{pmatrix}\sum_{i=1}^{l}x_{i}\\x_{1},...,x_{l}\end{pmatrix}\sim\sum a_{i_{1},...,i_{l}}\Pi_{j=1}^{l}x_{j}^{i_{j}}+o(1)$$

$$\begin{Bmatrix}
\sum_{i=1}^{l}x_{i}\\x_{1},...,x_{l}
\end{Bmatrix}= \frac{l^{\sum_{i=1}^{l}x_{i}+\frac{l}{2}}}{\sqrt{2\pi \sum_{i=1}^{l}x_{i}}^{l-1}}e^{-\frac{l}{2\sum_{i=1}^{l}x_{i}}(\sum_{i=1}^{l}(x_{i}-\frac{\sum_{i=1}^{l}x_{i}}{l})^{2})+o(1)}$$

This implies that should $\sum_{i=1}^{l}x_{i}$ be some constant large number, then its Schwartz class \cite{1972iv}, so all derivatives of it have Fourier transforms which are Schwartz class.

\end{theorem}
\begin{proof}

Let $a_{i_{1},...,i_{n}}$ be the taylor coefficients of the asymptotic function in Theorem~\ref{thm:multiNomBound}. Then:

$$lim_{m\rightarrow\infty}\mathcal{T}_{cont}^{m}\begin{pmatrix}\sum_{i=1}^{l}[m x_{i}] \\ [m x_{1}],...,[m x_{l}] \end{pmatrix}=lim_{m\rightarrow\infty}\sum a_{i_{1},...,i_{l}}\Pi_{j=1}^{l}([m x_{j}]+o(m))^{i_{j}}m^{\sum_{j=1}^{l}i_{j}}$$

This is due to the fact that from \cite{continuouslatticepath}, the paths with $n$ linear segments in the continuous multinomial correspond to taylor series coefficients with a total mixed degree of $n$. The asymptotic function above clearly indicates where terms of fixed power in $x_{i}$ are going to be located, allowing us to properly place the powers $m^{n}$ dictated by the operator $T_{cont}^{m}$.

Splitting the factor of m and using Theorem~\ref{thm:contconv} for well-behavedness to push the limit past the sum, this becomes:

$$\sum a_{i_{1},...,i_{l}}\Pi_{j=1}^{l}(lim_{m\rightarrow\infty}\frac{[m x_{i}]}{m}+o(1))^{i_{j}}=\sum a_{i_{1},...,i_{l}}\Pi_{j=1}^{l}(x_{i}+o(1))^{i_{j}}$$

This yields the theorem.

\end{proof}
\begin{theorem}
\label{thm:axes}

Let $d=2$. Then one $\mathcal{A}_{n}$ is the set $\{\frac{1}{gcd(x,I,t)}(x,t)\in\mathbb{N}^{2}|x^{2}+I^{2}=t^{2},(x,I,t)\in\mathbb{Z}^{3},0\le t\le n\}\cup\{(\pm 1,1)\}$ (i.e. a minimal $\mathcal{A}^{gen}$ for $d_{n}$ is this above set).

\end{theorem}
\begin{proof}

Consider the slopes of the half-plane boundaries which form $d_{n}(0,\vec{x})=0$ in Eq~\ref{eq:polygon}. The smallest magnitude slope is 1 (corresponding to $k=n-1$); if we had smaller slopes, $\mathcal{A}_{n}$ would intersect non-trivially with non-local paths (a contradiction). For $\vec{v}\in \mathbb{R}^{2}$, there is some R sufficiently such that if $proj_{x}(\vec{v})>R$, then because $1$ is the smallest slope, the equation for it must be the minimum of $d_{n}$. So, $\pm proj_{x}(\vec{v})=proj_{t}(\vec{x})$. But, we can take any vector in $\{d_{n}(0,\vec{x})=0\}$ and multiply its magnitude by $R$ and see that it necessarily must be parallel to $(\pm 1,1)$. Therefore, $(\pm 1,1)$ generates all null paths. Furthermore, the arguments of Theorem~\ref{thm:axes} show that $(\pm 1,1)$ cannot be generated except by themselves and must then lie in $\mathcal{A}_{n}$. So, we now need only need concentrate on non-null elements in $\mathcal{A}_{n}$ (i.e. $d_{n}(0,\vec{x})>0$).

\begin{figure}[h]
    \centering
    \includegraphics[width = 6 cm]{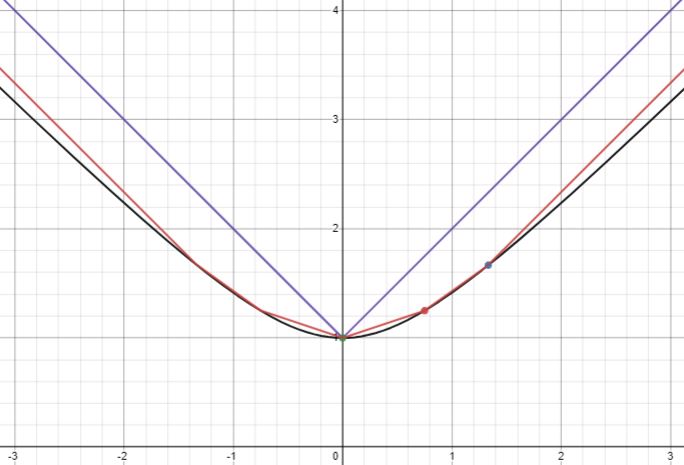}
    \caption{The sets $\mathcal{C}_{n}=\{d_{n}(0,\vec{x})=1\}$ for $n=\{1,2,3,4\}$ in blue, $n=5$ in red, and $C_{l_{2}^{*}}=\{d_{l_{2}^{*}}(0,\vec{x})=1\}$ in black.}
    \label{fig:CloserCirc}
\end{figure}

Let $d_{n}(0,\vec{x})=d_{l_{2}^{*}}(0,\vec{x})=r$. This yields $d_{n}(0,\frac{\vec{x}}{r})=d_{l_{2}{*}}(0,\frac{\vec{x}}{r})=1$. If we define $\mathcal{C}_{n}=\{d_{n}(0,\vec{x})=1\}$ and $\mathcal{C}_{l_{2}}=\{d_{l_{2}^{*}}(0,\vec{x})=1\}$, then the aforementioned set is contained within $\mathcal{C}_{n}\cap\mathcal{C}_{l_{2}^{*}}$. These objects are all plotted in Figure~\ref{fig:CloserCirc}. If we look at the argument of the minimum defining $d_{n}$ in Equation~\ref{eq:polygon}, you will see implicit equations for some lines. Namely, if we take that argument and set it equal to 1, we obtain Equation~\ref{eq:settoone}.

\begin{equation}
    \label{eq:settoone}
    \frac{|proj_{x}(\vec{a}_{i+1}-\vec{a}_{i})proj_{t}(\vec{x}_{2}-\vec{x}_{1})|-|proj_{t}(\vec{a}_{i+1}-\vec{a}_{i})proj_{x}(\vec{x}_{2}-\vec{x}_{1})|}{|proj_{x}(\vec{a}_{i+1}-\vec{a}_{i})proj_{t}(\vec{a}_{i})|-|proj_{t}(\vec{a}_{i+1}-\vec{a}_{i})proj_{x}(\vec{a}_{i})|}=1
\end{equation}

We introduce in Section~\ref{section:intro} the points in $\{(\frac{x}{I},\frac{t}{I})\in\mathbb{Q}^{2}|x^{2}+I^{2}=t^{2},(x,I,t)\in\mathbb{Z}^{3},0\le t\le n\}$ as the normalized Pythagorean triples with hypotenuse less than n. If we index this set by order of their x coordinate from most to least negative, then Equation~\ref{eq:settoone}'s linear equations are precisely those lines adjoining $\vec{a}_{i}$ to $\vec{a}_{i+1}$ as points in $\mathbb{R}^{2}$. If the minimum of these lines equals 1, it implies one of these lines equals, so $\mathcal{C}_{n}$ is composed of piecewise linear segments that adjoin $\vec{a}_{i}$ and $\vec{a}_{i+1}$. Every point above one of these lines satisfies Equation~\ref{eq:settoone} where its equal to $r\in (1,\infty)$ instead of 1. This fact, along with the monotonically increasing slope between $\{\vec{a}_{i},\vec{a}_{i+1}\}$ as points on $\mathbb{C}_{l_{2}^{*}}$, implies that between $(proj_{x}(\vec{a}_{i}),proj_{x}(\vec{a}_{i+1}))\subset\mathbb{R}$ the linear segment that composes $\mathcal{C}_{n}$ is precisely the line adjoining $\vec{a}_{i}$ and $\vec{a}_{i+1}$. This is again illustrated in Figure~\ref{fig:CloserCirc}. Since $\mathcal{C}_{l_{2}^{*}}$ is concave up, we have $\mathcal{C}_{l_{2}}\cap\mathcal{C}_{n}=\{(\frac{x}{I},\frac{t}{I})\in\mathbb{Q}^{2}|x^{2}+I^{2}=t^{2},(x,I,t)\in\mathbb{Z}^{3},0\le t\le n\}$.

The above fact shows us $\mathcal{A}^{all}$ are vectors parallel to the set $\mathcal{A}_{n}$ in the hypothesis of this Theorem. Here, we employ some facts about primitive Pythagorean triples \cite{rational}: if we divide $(x,I,t)$, a Pythagorean triple, by the gcd of the three elements, then the resulting object is a primitive Pythagorean triple. These primitive pythagorean triples build up all pythagorean triples along their direction and by extension all $\mathcal{A}^{all}$. They also all lie in different directions, so they constitute $\mathcal{A}^{n}$ as defined in Section~\ref{section:notation}.

\end{proof}

\begin{theorem}

Let $d\ge 2$. Then $\mathcal{A}_{1}=\{(\pm e_{1},1),...,(\pm e_{d},1),(0,1)\}$ where $e_{i}$ denote unit directions in $\mathbb{Z}^{d}$

\label{thm:axes2}

\end{theorem}

\begin{proof}

First if $d_{l_{1}}(0,\vec{x})=0$ then $proj_{t}(\vec{x})=\sum_{i}|proj_{x_{i}}(\vec{x})|$. This implies that the sum of the absolute value of the spatial coordinates equals $proj_{t}(\vec{x})$. Therefore, it may be generated by adding $proj_{t}(\vec{x})$ elements of $\{(\pm e_{1},1),...,(\pm e_{d},1)\}$. This set (except in some tropical examples) lies in $\mathcal{A}$. Any $\mathcal{A}^{gen}\subset \mathcal{A}^{all}$ would need to contain all of these elements as if $\sum_{i}\vec{a}_{i}=(e_{1},1)$ for $\{\vec{a}\}_{i=1}^{N}\subset\mathcal{A}$ then $\sum_{i}proj_{t}(\vec{a}_{i})=1$ and all except one $\vec{a}_{i}$ is zero. So, we have obtained a minimal set required to generate null vectors. Say $\vec{x}\in X$ such that $\vec{x}-proj_{t}(\vec{x})\hat{t}=0$. Then $d_{l_{1}}(0,\vec{x})=d_{l_{2}^{*}}=proj_{t}(\vec{x})$ and furthermore its generated by $(0,1)$. Again, any $A^{gen}$ would need to contain $(0,1)$. So, if we show that these above vectors are the only ones upon which $d_{l_{1}}$ and $d_{l_{2}^{*}}$ agree, then we are done. We note that the surface $\{d_{1}(0,\vec{x})=1\}$ is a polyhedral cone which touches $\{d_{l_{2}^{*}}(0,\vec{x})=1\}$ atleast at $(\vec{0},1)$. Let $proj_{t}(\vec{x})=T\in [1,\infty)$. Then, $\{d_{l_{2}^{*}}(0,\vec{x})=1\}$ are rational points on a sphere in $\mathbb{Z}^{d}$ of radius $\sqrt{T^{2}-1}$. Meanwhile, $\{d_{n}(0,\vec{x})=1\}$ will be a polyhedral shape with maximum distance from the origin of $T-1$ (as the maximum distances are aligned with the axes in the $l_{1}$ sphere). It is immediate that $T-1<\sqrt{T^{2}-1}$ for $T\in (1,\infty)$, and so our $d_{1}$ 'sphere' lies above the $d_{l_{2}^{*}}$ in the upper half space at all points except $(0,1)$. This implies, by scaling, that the only non-null point upon which the metrics agree is multiples of $(0,1)$.

%polyhedron bound in each coordinate $x_{i}$ by a path moving solely in said direction).

%Furthermore $\{d_{l_{2}^{*}}(0,\vec{x})=1\}$ is smooth on the upper half space, and has vanishing gradient at 0. The other space has discontinuous, instantly non vanishing gradient as the intersection of half spaces. Therefore locally it lies above $\{d_{l_{2}^{*}}(0,\vec{x})=1\}$ in the t-th coordinate for some neighborhood about $(\vec{0},1)$. Finally $\{d_{l_{2}^{*}}(0,\vec{x})=1\}$ has no other critical points and constant curvature. If we consider $\{d_{l_{2}^{*}}(0,\vec{x})=1|\{d_{l_{2}}(\vec{x}-proj_{t}(\vec{x}))>R\}\}$

%If they agree elsewhere, then it would need to be a non-null vector. %there is some $\vec{v}\in \mathbb{Z}^{d}$ of integer length 

%Now let $d_{l_{1}}(0,\vec{x})=d_{l_{2}^{*}}(0,\vec{x})>0$

%together without  that the smallest $|*|_{l_{2}}$ norm a non-zero vector in $\mathbb{Z}^{d}$ can take is 1. Furthermore it only takes it iff one coordinate is $\pm 1$. 

\end{proof}
\begin{theorem}
\label{thm:generate}

Let $\gamma=\{x_{i}\}_{i=1}^{n}\in \Gamma_{n}$. Let $\mathcal{A}_{n}$ denote some axes of symmetry of $d_{n}$. Then, we can choose our $\{x_{i}\}^{n}$ such that $x_{i+1}-x_{i}\in\mathcal{A}_{n}$. This alteration will yield the same path up to our equivalence relation on $\Gamma$, and this equivalence class representative is unique.

\end{theorem}

\begin{proof}

Let $\gamma\in \Gamma_{n}$ such that there is some index $i$ where $x_{i+1}-x_{i}$ is not in $\mathcal{A}_{n}$. We know $d_{l_{2}^{*}}(\vec{x}_{i+1},\vec{x}_{i})=d_{n}(\vec{x}_{i+1},\vec{x}_{i})$ so by the definition of $\mathcal{A}_{n}$  we know $\vec{x}_{i+1}-\vec{x}_{i}=\sum_{i}^{N}\vec{a}_{i}$ for $\{\vec{a}_{i}\}_{i=1}^{N}\subset \mathcal{A}_{n}$. Let $\gamma'$ be defined such that $x'_{j}=x_{j}$ for $j\le i$, that $x'_{i+k}=x_{j}+\vec{a}_{i}$ for $k<n$, that $x'_{j+n}=x_{j+1}$ for $j\ge i$. We just replaced our single difference sequence element with a finite number of difference sequence elements, all of which were in $\mathcal{A}_{n}$. This portion of the path all lies on the same points in $X$; therefore, $\gamma'\sim\gamma$ under our equivalence relation.

If we take $\gamma\in \Gamma_{n}$, by a finite number of recursions of the above argument, we obtain a $\gamma'\in \Gamma_{n}$ such that $\gamma'\sim \gamma$ and its composed only of difference sequences among $\mathcal{A}_{n}$. We can proceed by induction to show this representation of the difference sequence of $\gamma$ among $\mathcal{A}_{n}$. If $x_{1}-x_{0},x_{1}'-x_{0}'\in \mathcal{A}_{n}$ are not equal, then they are not directed in the same direction. So, they are locally not the same piecewise linear graph. This is a contradiction. We proceed by finite induction and show that the difference sequence is the same.

\end{proof}

\begin{theorem}
    \label{thm:tripledensity}
    Denote the set of primitive pythagorean triples with hypotenuse below $n\in \mathbb{R}$ as $\mathcal{A}_{n}$. Then, $\mathcal{A}_{n}=\frac{n}{2\pi}+O(\sqrt{n}ln(n))$, and they are equidistributed on the unit circle when ordered according to hypotenuse (as the hypotenuse goes to infinity).
\end{theorem}

\begin{proof}

    The proof of this theorem can be found in \cite{Takloo-Bighash2018} on pgs. 217 and 242.
    
\end{proof}

\begin{theorem}
    \label{thm:properties}
    The continuum multinomial has the following property
    $$\begin{Bmatrix}\sum_{i}x_{i}\\x_{1},...,x_{n}\end{Bmatrix}=\int_{-\infty}^{\infty}\begin{Bmatrix}\sum_{i}x_{i}\\x_{1},...,I\end{Bmatrix}\begin{Bmatrix}I\\x_{n-1},x_{n}\end{Bmatrix}dI$$

    This implies that should any of the coefficients of the continuum multinomial coefficient be zero, than the function is also zero.
\end{theorem}
\begin{proof}
    The first property of the continuum multinomial is established by Theorem~\ref{thm:disctocont} and the proof is performed in the same manner as Theorem~\ref{thm:conttails}. Namely the multiplicative property described is obtained in the discrete setting, and commutes past limits to apply to the continuum multinomial. The second property is established by an explicit equation found in \cite{continuousbin}. Cano and Diaz found that their version of the binomial coefficient has the property $\begin{Bmatrix}x\\x\end{Bmatrix}=\begin{Bmatrix}x\\0\end{Bmatrix}=x+2$; this was obtain from the explicit expression $\begin{Bmatrix}x\\s\end{Bmatrix}=\sum_{n=0}^{\infty}(x+2n+2)\frac{s^{n}(x-s)^{n}}{n!(n+1)!}$. Now our expression for the continuum multinomial coefficient only differs in that every power $(x-s)^{i}s^{j}$ has a multiplicative factor $\sqrt{i}\sqrt{j}$ placed next to it. We can rewrite Cano's expression in terms of these powers like so: 
    
    $\sum_{n=0}^{\infty}(x-s+s+2n+2)\frac{s^{n}(x-s)^{n}}{n!(n+1)!}=\sum_{n=0}^{\infty}\frac{s^{n+1}(x-s)^{n}}{n!(n+1)!}+\frac{s^{n}(x-s)^{n+1}}{n!(n+1)!}+(2n+2)\frac{s^{n}(x-s)^{n}}{n!(n+1)!}$ 
    
    meaning we obtain the following expression for the continuum multinomial coefficient

    \begin{equation}
        \label{eq:multhelpt}
        \begin{Bmatrix}x\\s\end{Bmatrix}=\sum_{n=0}^{\infty}\frac{\sqrt{n+1}\sqrt{n}s^{n+1}(x-s)^{n}}{n!(n+1)!}+\frac{\sqrt{n}\sqrt{n+1}s^{n}(x-s)^{n+1}}{n!(n+1)!}+(2n+2)n\frac{s^{n}(x-s)^{n}}{n!(n+1)!}
    \end{equation}

    Let $x=s$. In the first term of Equation~\ref{eq:multhelpt} $x-s$ goes to zero unless $n=0$, in which case our multplicative factor takes it to zero. In the next expression it will always be zero, and in the third portion of Equation~\ref{eq:multhelpt} we return to the first case. So the binomial coefficient evalautes to zero when any of its coefficients do. This, along with the multiplicative properties of the continuum multinomial already established, obtains the desired result.
\end{proof}

\section*{Acknowledgement}

I want to thank Eviatar Procaccia and Parker Duncan from Technion University for their help in developing this paper through work in \cite{duncan2020elementary} and \cite{duncan2021discrete}. Hailey Leclerc proved invaluable in editing this document. I would also like to thank Beba. I miss you immensely.

\section{Data Availibility Statement}

The manuscript has no associated data.

\bibliographystyle{amsplain}
\bibliography{main.bib}

\providecommand{\bysame}{\leavevmode\hbox to3em{\hrulefill}\thinspace}
\providecommand{\MR}{\relax\ifhmode\unskip\space\fi MR }
\providecommand{\MRhref}[2]{%
  \href{http://www.ams.org/mathscinet-getitem?mr=#1}{#2}
}
\providecommand{\href}[2]{#2}
\begin{thebibliography}{10}

\bibitem{Albeverio2009}
Sergio Albeverio and Sonia Mazzucchi, \emph{A survey on mathematical feynman
  path integrals: Construction, asymptotics, applications}, pp.~49--66,
  Birkh{\"a}user Basel, Basel, 2009.

\bibitem{continuousbin}
Leonardo Cano and Rafael Diaz, \emph{Continuous analogues for the binomial
  coefficients and the catalan numbers}, 2016.

\bibitem{duncan2020elementary}
Parker Duncan, Rory O'Dwyer, and Eviatar~B. Procaccia, \emph{An elementary
  proof for the double bubble problem in $\ell^1$ norm}, 2020.

\bibitem{duncan2021discrete}
\bysame, \emph{Discrete $\ell^{1}$ double bubble solution is at most ceiling
  $+2$ of the continuous solution}, 2021.

\bibitem{feynman_hibbs_styer_2017}
Richard~Phillips Feynman, Albert~R. Hibbs, and Daniel~F. Styer, \emph{Quantum
  mechanics and path integrals}, Dover Publications, 2017.

\bibitem{folland}
Gerald~B. Folland, \emph{Real analysis: Modern tichniques and their
  applications}, Wiler-Interscience Publication, 1994.

\bibitem{latticeProb}
Manuel Friedrich, Wojciech Górny, and Ulisse Stefanelli, \emph{The
  double-bubble problem on the square lattice}, 2021.

\bibitem{glimm_jaffe_1987}
James Glimm and Arthur Jaffe, \emph{Quantum physics: A functional integral
  point of view}, Springer, 1987.

\bibitem{Hong_Hao_2010}
Zhang Hong-Hao, Feng Kai-Xi, Qiu Si-Wei, Zhao An, and Li~Xue-Song, \emph{On
  analytic formulas of feynman propagators in position space}, Chinese Physics
  C \textbf{34} (2010), no.~10, 1576--1582.

\bibitem{344669}
Iosif~Pinelis (https://mathoverflow.net/users/36721/iosif pinelis),
  \emph{Asymptotics of multinomial coefficients}, MathOverflow,
  URL:https://mathoverflow.net/q/344669 (version: 2022-01-02).

\bibitem{https://doi.org/10.48550/arxiv.1101.1024}
Nam-Gyu Kang and Nikolai Makarov, \emph{Gaussian free field and conformal field
  theory}, 2011.

\bibitem{Lewin_2022}
Mathieu Lewin, \emph{Coulomb and riesz gases: The known and the unknown},
  Journal of Mathematical Physics \textbf{63} (2022), no.~6, 061101.

\bibitem{linde_2017}
A.~Linde, \emph{Particle physics and inflationary cosmology}, CRC Press, 2017.

\bibitem{tropicalGeometry}
Ralph Morrison, \emph{Tropical geometry}, 2019.

\bibitem{pesky}
Michael~Edward Peskin and Daniel~V. Schroeder, \emph{{An Introduction to
  Quantum Field Theory}}, Westview Press, 1995, Reading, USA: Addison-Wesley
  (1995) 842 p.

\bibitem{1972iv}
Micheal Reed, \emph{Methods of modern mathematical physics}, Academic Press,
  1972.

\bibitem{Takloo-Bighash2018}
Ramin Takloo-Bighash, \emph{A pythagorean introduction to number theory: Right
  triangles, sums of squares, and arithmetic}, pp.~3--12, Springer
  International Publishing, Cham, 2018.

\bibitem{rational}
Lin Tan, \emph{The group of rational points on the unit circle}, Mathematics
  Magazine \textbf{69} (1996), no.~3, 163--171.

\bibitem{vafa2005string}
Cumrun Vafa, \emph{The string landscape and the swampland}, 2005.

\bibitem{vigneaux2020homological}
Juan~Pablo Vigneaux, \emph{A homological characterization of generalized
  multinomial coefficients related to the entropic chain rule}, 2020.

\bibitem{continuouslatticepath}
T.~Wakhare, C.~Vignat, Q.~N. Le, and S.~Robins, \emph{A continuous analogue of
  lattice path enumeration}, 2017.

\bibitem{zwiebach_2004}
Barton Zwiebach, \emph{A first course in string theory}, Cambridge University
  Press, 2004.

\end{thebibliography}

\end{document}